\documentclass[letterpaper]{comjnl}

\usepackage{graphicx}
\usepackage{cite}
\usepackage{amsmath}
\usepackage{amssymb}
\usepackage{array}
\usepackage{color, xcolor}
\usepackage{cite}
\usepackage{multirow}
\usepackage[flushleft]{threeparttable}

\usepackage[T1]{fontenc}

\usepackage{aecompl}

\newcommand\blfootnote[1]{%
	\begingroup
	\renewcommand\thefootnote{}\footnote{#1}%
	\addtocounter{footnote}{-1}%
	\endgroup
}

\begin{document}
\title[Non-Backtracking Centrality Based Random Walk on Networks]{Non-Backtracking Centrality Based Random Walk on Networks}
\author{Yuan~Lin\thanks{}}
\author{Zhongzhi~Zhang}\email{zhangzz@fudan.edu.cn}
\affiliation{Shanghai Key Laboratory of Intelligent Information
	Processing, School of Computer Science, Fudan University, Shanghai 200433, China}
\shortauthors{Y. Lin and Z. Zhang}


\begin{abstract}
Random walks are a fundamental tool for analyzing realistic complex networked systems and implementing randomized algorithms to solve diverse problems such as searching and sampling. For many real applications, their actual effect and convenience depend on the properties (e.g. stationary distribution and hitting time) of random walks, with biased random walks often outperforming traditional unbiased random walks (TURW). In this paper, we present a new class of biased random walks, non-backtracking centrality based random walks (NBCRW) on a network, where the walker prefers to jump to neighbors with high non-backtracking centrality that has some advantages over eigenvector centrality. We study some properties of the non-backtracking matrix of a network, on the basis of which we propose a theoretical framework for fast computation of the transition probabilities, stationary distribution, and hitting times for NBCRW on the network. Within the paradigm, we study NBCRW on some model and real networks and compare the results with those corresponding to TURW and maximal entropy random walks (MERW), with the latter being biased random walks based on eigenvector centrality. We show that the behaviors of stationary distribution and hitting times for NBCRW widely differ from those associated with TURW and MERW, especially for heterogeneous networks.
\end{abstract}

\maketitle

\section{Introduction}

As a fundamental and powerful tool, random walks have found a wide range of applications in computer science and engineering. For example, in the area of communication and information networks, random walks can not only model and describe information delivery~\cite{ChBa11} and data gathering~\cite{ZhYaTiGaWaXi15, LeKw16}, but also quantify and predict the throughput~\cite{ElMaPr06, LiJiNiKa12}, latency performance~\cite{ChBa11}, transition~\cite{LiZh13IEEE} and search costs~\cite{BeQuBa09, LiLiWaCh09}. Other related applications of random walks in computer science include community detection~\cite{PoLa05}, recommendation system~\cite{FoPiReSa07}, computer vision~\cite{GoHuRa09}, image segmentation~\cite{Gr06}, sampling networks~\cite{RiTo10,RiWaMuTo12}, to name a few. The statistical properties of random walks play an important role in their applications, since they not only characterize the behavior of random walks themselves, but also capture the performance metrics of different applications. For example, stationary probability of stationary distribution can measure the node importance~\cite{BrPa98} of a network, as well as the visual saliency at a location~\cite{YuZhTiTa14}, while hitting time can serve as search performance gauge~\cite{BeQuBa09}. Thus, the properties of random walks have a strong impact on, even determine to a large extent, the effects of their applications.

\blfootnote{*Currently at: Institute for Theoretical Computer Science, School of Information Management and Engineering, Shanghai University of Finance and Economics, Shanghai 200433, China.}




Among various random walks, the traditional unbiased random walk (TURW) is probably the simplest one, where the transition probability from the current location to any neighbor at next time step is uniform. Nevertheless, a vast majority of real-life networks are heterogeneous~\cite{Ne03}, implying that the importance or role of different nodes are also distinct. Thus, random walks in realistic heterogeneous networks should be biased~\cite{Be09,GjKuBuMa10}, with transition probability to an important neighbor higher than that of an ordinary neighbor. A lot of works show that in comparison with TURW, biased random walks are superior in some concrete applications, e.g., network search~\cite{Be09,IkKuYa09} sampling~\cite{MaBe11}. A typical biased random walk is maximal entropy random walk (MERW)~\cite{Pa64}, which has received considerable attention~\cite{BuDuLuWa09,GoLa08,PeZh14,LiZh14}. Entropy of random walks quantifies the randomness of trajectories and can measure mobility of random walker~\cite{KaGrTh13}. MERW displays some remarkable properties different from those of TURW, e.g. small relaxation time~\cite{OcBu12}, localization of stationary distribution~\cite{BuDuLuWa09}. In the past years, MERW has been applied to several aspects, such as link prediction~\cite{LiYuLi11}, visual saliency~\cite{YuZhTiTa14} and digital image forensics~\cite{KoHu16}, and produced more desirable effects. 

MERW is in fact a biased random walk with transition biasing towards neighboring nodes with high eigenvector centrality~\cite{Bo72}, i.e. principal eigenvector of adjacent matrix. However, a recent research~\cite{MaZhNe14} pointed out that standard centrality undergoes a localization transition in heterogeneous networks, which leads to most of weight concentrating around the hub node and its vicinity. Thus, as a common measure of node importance, the standard eigenvector centrality fails to discriminate those nodes with small weight. As a remedy, an alternative centrality measure, non-backtracking centrality, was proposed~\cite{MaZhNe14}, which reserves the advantage of standard centrality but avoids its deficiency. This new centrality measure is based on non-backtracking matrix~\cite{Ha89, KrMoMoNeSlZdZh13}, which has been successfully applied to many aspects, such as community detection~\cite{KrMoMoNeSlZdZh13}, percolation~\cite{KaNeZd14,LiChZh17}, epidemic spreading~\cite{ShScMo15} and identifying influential nodes ~\cite{MoMa15}. Since the node properties, based on which the walker has preference to jump towards different nodes, play a central role in determining the behavior of biased random walks, an interesting question arises naturally: How does a random walk behave if non-backtracking centrality is incorporated into its transition probabilities? 


In this paper, we design a new biased random walk, Non-Backtracking Centrality based Random Walk (NBCRW), with the transition probabilities dependent on the non-backtracking centrality. We present a framework for computing quickly transition probabilities, stationary distribution, and hitting times of NBCRW, and provide analytical expressions for stationary distribution and hitting times. Within this framework, we study NBCRW on some synthetic and real networks, and compare their results with those with respect to TURW and MERW. We show that the behaviors of NBCRW differ greatly from those of TURW and MERW, in particular for heterogeneous networks.



The main contributions of this paper are summarized as follows.

\begin{itemize}

\item We propose a novel type of biased random walks, non-backtracking centrality based random walks (NBCRW), in which the transition probability is proportional to the non-backtracking centrality.

\item We develop a theoretical framework for efficiently computing transition probabilities of NBCRW as well as its properties, including stationary distribution and hitting times. We derive an analytical expression of the stationary distribution of NBCRW in terms of the leading eigenvalue of non-backtracking matrix and non-backtracking centrality. We also determine hitting times for NBCRW, including hitting time from an arbitrary node to another one, partial mean hitting time to a given target, and global mean hitting time to a uniformly selected node.

\item Within the established general framework, we study analytically or numerically NBCRW in model and realistic networks, and compare the results with those corresponding to TURW and MERW. We show that the stationary distribution and hitting times behave differently from those of TURW and MERW.

\end{itemize}

The remainder of this paper is organized as follows. Section~\ref{Sec2} presents a brief introduction to networks and an overview of TURW and MERW on networks. Section~\ref{Sec3} is devoted to the formulation of NBCRW. Section~\ref{Sec4} gives the experiment results and comparison between NBCRW, TURW and MERW in model and real-life networks. Section~\ref{SecS} reports the exact analytical results of stationary distributions and hitting times for NBCRW, TURW and MERW in a class of rose graphs. Section~\ref{Sec5} concludes the paper.


\section{Preliminaries}~\label{Sec2}


In this section, we introduce some useful concepts for graphs and discrete-time random walks on graphs.

\subsection{Concepts for Graphs and Random Walks}

Let $\mathcal{G}(\mathcal{V}, \mathcal{E})$ be a finite connected undirected network (graph) of $N$ nodes and $E$ edges,  with node set $\mathcal{V}=\{1, 2, \cdots, N\}$ and edge set $\mathcal{E}=\{(i, j)|i, j\in \mathcal{V}\}$. The connectivity of nodes is defined by the adjacency matrix ${\bf A}=(a_{ij})_{N \times N}$, in which the element $a_{ij}=1$ if $(i, j)\in \mathcal{E}$, and $a_{ij}=0$ otherwise. Let $\mathcal{N}_i$ denote the set of neighbors of node $i$. The degree of node $i$ is $d_i=|\mathcal{N}_i|=\sum_{j=1}^{N}a_{ij}$, which is $i$th nonzero entry of the diagonal degree matrix ${\bf D}={\rm diag}(d_1, d_2, \cdots, d_N)$. The Laplacian matrix of $\mathcal{G}$ is defined to be ${\bf L}={\bf D}-{\bf A}$.

For a graph $\mathcal{G}$, we can define a discrete time nearest-neighbor random walk taking place on it. Any random walk on a network $\mathcal{G}$ is in fact a Markov chain characterized by a unique stochastic matrix ${\bf P}=(p_{ij})_{N \times N}$, also called transition probability matrix, with entry $p_{ij}$ describing the transition probability from node $i$ to a neighboring node $j$. 


\begin{definition}
	For an irreducible random walk on graph $\mathcal{G}$, the stationary distribution $\pi=(\pi_1, \pi_2, \cdots, \pi_N)$ is an $N$-dimension vector satisfying $\pi{\bf P}=\pi$ and $\sum_{i=1}^N \pi_i=1 $.
	\end{definition}

The stationary probabilities of stationary distribution can be employed to rank nodes in a network~\cite{BrPa98}. 

Another fundamental quantity relevant to random walks is hitting time~\cite{CoBeTeVoKl07}.
\begin{definition}
For a random walk on graph $\mathcal{G}$, the hitting time from node $i$ to node $j$ ($j\neq i$), denoted by $T_{ij}$, stands for the expected jumping steps required for the walker starting from the source node $i$ to arrive at the target node $j$ for the first time.
\end{definition}

The hitting time is a significant indicator to measure the transition or research cost in a network~\cite{BeQuBa09}. 
Based on hitting time, we can further define some other quantities for random walks, such as partial mean hitting time and global mean hitting time.

\begin{definition}
For a random walk on graph $\mathcal{G}$, the partial mean hitting time to node $j$, denoted by $T_j$, is the average of hitting times $T_{ij}$ over all source nodes in the network:
\begin{equation}\label{C1}
T_j=\frac{1}{N-1}\sum_{i=1}^{N}T_{ij}.
\end{equation}
\end{definition}

The partial mean hitting time $T_j$ is actually mean absorbing time of an absorbing Markov chain with $j$ being the absorbing state, reflecting the absorbing efficiency of node $j$~\cite{LiJuZh12,ErGoSaSe14}. It was recently utilized to measure the importance of node $j$, and is thus called Markov centrality~\cite{WhSm03}.

\begin{definition}
For a random walk on graph $\mathcal{G}$, the global mean hitting time, denoted by $\langle T\rangle$, is the average of hitting times $T_{ij}$ over all $N(N-1)$ pair of nodes, equivalent to the mean hitting time to a uniform distributed node, which is given by
\begin{equation}\label{C2}
\langle T\rangle=\frac{1}{N(N-1)}\sum_{i=1}^{N}\sum_{j\neq i}T_{ij}=\frac{1}{N}\sum_{j=1}^{N}T_j.
\end{equation}
\end{definition}

The global mean hitting time can be applied to gauge the search efficiency of a network~\cite{FeQuYi14}.


Given a network, we can define different random walks. Below we only introduce two much studied random walks: traditional unbiased random walk (TURW) and maximal entropy random walk (MERW).


\subsection{Traditional Unbiased Random Walk}

For TURW on graph $\mathcal{G}$, the transition probability from a node $i$ to one of its neighboring nodes $j$ is identical, namely
\begin{equation}\label{B2}
p_{ij}=\frac{a_{ij}}{d_{i}}.
\end{equation}
Thus, the transition probability matrix is ${\bf P}={\bf D}^{-1}{\bf A}$, and the stationary distribution is~\cite{Lo96, AlFi99}
\begin{equation}\label{B3}
\pi^{\rm T}=(\pi_1^{\rm T}, \pi_2^{\rm T}, \cdots, \pi_N^{\rm T})=\left(\frac{d_1}{2E}, \frac{d_2}{2E}, \cdots, \frac{d_N}{2E}\right),
\end{equation}
which implies that all nodes with the same degree have identical occupation probability in the stationary state.

The hitting time for TURW on $\mathcal{G}$ can be expressed in terms of spectra of its Laplacian matrix ${\bf L}$. Let $0=\sigma_1<\sigma_2\leq \cdots \leq \sigma_N$ be the $N$ eigenvalues of ${\bf L}$, and let $\mu_1, \mu_2, \cdots, \mu_N$ be their corresponding normalized mutually orthogonal eigenvectors, where $\mu_i=(\mu_{i1}, \mu_{i2}, \cdots, \mu_{iN})^{\top}$ for each $i=1,2, \cdots, N$. Then, the hitting time $T_{ij}$,  partial mean hitting time, and global mean hitting time can be represented by~\cite{LiJuZh12}
\begin{equation}\label{B31}
T_{ij}=\sum_{z=1}^{N}d_z\sum_{k=2}^{N}\frac{1}{\sigma_k}(\mu_{ki}\mu_{kz}-\mu_{ki}\mu_{kj}-\mu_{kj}\mu_{kz}+\mu_{kj}^2),
\end{equation}
\begin{equation}\label{B32}
T_{j}=\frac{N}{N-1}\sum_{k=2}^{N}\frac{1}{\sigma_k}\left(2E\times\mu_{kj}^2-\mu_{kj}\sum_{z=1}^{N}d_z\mu_{kz}\right)
\end{equation}
and
\begin{equation}\label{B33}
\langle T\rangle=\frac{2E}{N-1}\sum_{k=2}^{N}\frac{1}{\sigma_k},
\end{equation}
respectively.

\subsection{Maximal Entropy Random Walk}

Different from the TUWR, MERW on graph $\mathcal{G}$ is a biased random walk, whose transition probability is defined based on the leading eigenvalue and eigenvector of adjacency matrix $\bf A$. Let $\lambda_1>\lambda_2\geq\cdots\geq\lambda_N$ be the $N$ real eigenvalues of $\bf A$, and $\psi_1, \psi_2, \cdots, \psi_N$ their corresponding mutually orthogonal unit eigenvectors, where $\psi_i=(\psi_{i1}, \psi_{i2}, \cdots, \psi_{iN})^\top$ for each $i=1,2, \cdots, N$. Then, the transitional probability $p_{ij}$ from node $i$ to node $j$ in MERW is defined by~\cite{Pa64,BuDuLuWa09}
\begin{equation}\label{B4}
p_{ij}=\frac{a_{ij}}{\lambda_1}\frac{\psi_{1j}}{\psi_{1i}}.
\end{equation}
Note that principal eigenvector $\psi_1$ is in fact the frequently used centrality measure~\cite{Bo72}, with the entry $\psi_{1,i}$ defining a centrality score for node $i$. In this sense, MERW can be considered as a biased random walk based on eigenvector centrality.

Equation~(\ref{B4}) guarantees that MERW maximizes the entropy of a set of trajectories with a given length and end-nodes, leading to the maximal entropy rate of such process~\cite{BuDuLuWa09}. The stationary distribution of MERW is
\begin{equation}\label{B5}
\pi^{\rm M}=(\pi_1^{\rm M}, \pi_2^{\rm M}, \cdots, \pi_N^{\rm M})=(\psi_{11}^2, \psi_{12}^2, \cdots, \psi_{1N}^2).
\end{equation}
Since in some networks, especially heterogeneous networks, the eigenvector centrality  $\psi_1$ exhibits a localization phenomenon~\cite{MaZhNe14} with the weight of centrality concentrating around one or a few nodes with high degree in the networks,  from (\ref{B5}) one can see that in these networks, the stationary distribution for MERW displays a more evident localization transition: the abrupt focusing of occupation probabilities on just a few large-degree nodes and their neighbors. 





Interestingly, for MERW on graph $\mathcal{G}$, the hitting time $T_{ij}$, partial mean hitting time $T_{j}$, and global mean hitting time $\langle T\rangle$ can be expressed in terms of the eigenvalues and eigenvectors of adjacency matrix $\bf A$~\cite{LiZh14}:
\begin{equation}\label{B51}
T_{ij}=\frac{1}{\psi_{1j}^2}\sum_{k=2}^{N}\frac{\lambda_1}{\lambda_1-\lambda_k}\left(\psi_{kj}^2-\psi_{ki}\psi_{kj}\frac{\psi_{1j}}{\psi_{1i}}\right),
\end{equation}
\begin{equation}\label{B52}
T_j=\frac{1}{\psi_{1j}^2(N-1)}\sum_{k=2}^N\frac{\lambda_1}{\lambda_1-\lambda_k}\left(N\psi_{kj}^2-\psi_{kj}\psi_{1j}\sum_{i=1}^N\frac{\psi_{ki}}{\psi_{1i}}\right),
\end{equation}
\begin{align}\label{B53}
\langle T\rangle&=\frac{1}{N(N-1)}\sum_{j=1}^N\frac{1}{\psi_{1j}^2}\nonumber\\
&\quad\sum_{k=2}^N\frac{\lambda_1}{\lambda_1-\lambda_k}\left(N\psi_{kj}^2-\psi_{kj}\psi_{1j}\sum_{i=1}^N\frac{\psi_{ki}}{\psi_{1i}}\right).
\end{align}


\section{Formulation of Non-Backtracking Centrality Based Random Walk}~\label{Sec3}

For a biased random walk on a graph $\mathcal{G}$, its behavior depends on the property of the quantity with respect to nodes, based on which the transition probability is defined. As shown in a recent paper~\cite{MaZhNe14}, the eigenvector centrality has some flaws, e.g., localization transition, which results in obvious heterogeneity in the stationary distribution of MERW. Since non-backtracking centrality can avoid the deficiency of eigenvector centrality~\cite{MaZhNe14}, as a remedy of MERW, in this section, we propose a new biased random walk based on non-backtracking centrality. To begin with, we introduce the non-backtracking centrality and study some of its properties.

\subsection{Non-Backtracking Centrality}

The non-backtracking centrality~\cite{MaZhNe14} is defined and calculated by the Hashimoto or non-backtracking matrix~\cite{Ha89,KrMoMoNeSlZdZh13}, denoted by $\bf B$ that is a $2E\times 2E$ matrix. For any undirected network $\mathcal{G}$, we can transform it to a directed graph through replacing each undirected edge $(i,j)$ by two directed ones $i\to j$ and $j \to i$. The $2E\times 2E$ non-backtracking matrix $\bf B$ of $\mathcal{G}$ describes the relation between the $2E$ different directed edges, the element $B_{i\to j, k\to l}$ of which is defined as follows:
\begin{equation}\label{B6}
B_{i\to j, k\to l}=\left\{\begin{aligned}
&1, \quad j=k \quad{\rm and} \quad i\neq l, \\
&0, \quad {\rm otherwise}. \\
\end{aligned}
\right.
\end{equation}

Since all entries of the non-backtracking matrix $\bf B$ are non-negative real numbers, by the Perron-Frobenius theorem~\cite{St09}, its leading eigenvalue is real and non-negative, and there exists a corresponding leading eigenvector, whose elements are also non-negative real numbers. Let $\kappa$ be the leading eigenvalue of $\bf B$, and let $v_{i\to j}$ be the element of the leading eigenvector corresponding to the directed edge $i\to j$. Then, $v_{i\to j}$ represents the centrality of node $j$ neglecting any contribution from node $i$. According to the leading eigenvector of $\bf B$, one can define two centrality measures of each node~\cite{KrMoMoNeSlZdZh13}, outgoing centrality and incoming centrality,  by considering the outgoing and incoming edges of the node.
\begin{definition}
For a node $i$ in network $\mathcal{G}$, its outgoing centrality is
\begin{equation}\label{B7}
x_i=\sum_{j\in\mathcal{N}_i}v_{i\to j},
\end{equation}
and its incoming centrality is
\begin{equation}\label{B71}
y_i=\sum_{j\in\mathcal{N}_i}v_{j\to i}.
\end{equation}
\end{definition}
Note that the outgoing centrality $x_i$ is actually the non-backtracking centrality~\cite{MaZhNe14}.

\begin{lemma}
For a node $i$ in network $\mathcal{G}$, its outgoing and incoming centralities obey
\begin{equation}\label{S6}
\kappa y_i=(d_i-1)x_i.
\end{equation}
\end{lemma}
\begin{proof} By definition of eigenvalues and eigenvectors for matrix $\bf B$, we can establish equation
\begin{equation}\label{S2}
\kappa v_{i\to j}=\sum_{\substack{k\in\mathcal{N}_j\\k\neq i}}v_{j\to k} .
\end{equation}
Using (\ref{B7}) and (\ref{S2}), we rephrase (\ref{B71}) as
\begin{align}\label{S5}
y_i&=\frac{1}{\kappa}\sum_{j\in\mathcal{N}_i}\sum_{\substack{k\in\mathcal{N}_i\\k\neq j}}v_{i\to k}=\frac{1}{\kappa}\left(\sum_{j\in\mathcal{N}_i}\sum_{k\in\mathcal{N}_i}v_{i\to k}-\sum_{j\in\mathcal{N}_i}v_{i\to j}\right)\nonumber\\
&=\frac{1}{\kappa}\left(\sum_{j\in\mathcal{N}_i}x_i-x_i\right)=\frac{1}{\kappa}\left(d_i x_i-x_i\right),
\end{align}
which is equivalent to (\ref{S6}).
\end{proof}

If graph $\mathcal{G}$ is a tree, the leading eigenvalue $\kappa$ of its non-backtracking matrix ${\bf B}$ is zero. However, when $\mathcal{G}$ is not a tree, the leading eigenvalue $\kappa$ of ${\bf B}$ is positive, and the components of leading eigenvector  may be all non-negative. In what follows, we will consider the case when $\mathcal{G}$ are not trees.

For a network $\mathcal{G}$, computing its non-backtracking centrality involving computing the leading eigenvector of its non-backtracking matrix ${\bf B}$ of order $2E\times 2E$. If we directly compute the leading eigenvector according to definition, the time and space cost is very high. Fortunately, in practice, we can substantially reduce the consumption by executing a faster computation for $\kappa$ and non-backtracking centrality $x_i$, utilizing the Ihara determinant~\cite{Ha89,Ba92,AnFrHo07}.

\begin{lemma}\label{MandB}
For a network $\mathcal{G}$, its leading eigenvalue $\kappa$ of non-backtracking matrix $\bf B$ is equal to the leading eigenvalue of a $2N \times 2N$ matrix
\begin{equation}\label{M1}
{\bf M}=\left(\begin{array}{cc}{\bf A} & {\bf I}-{\bf D}\\{\bf I} & {\bf 0}\end{array}\right),
\end{equation}
where ${\bf I}$ is the $N\times N$ identity matrix. In addition, $x_1, x_2, \cdots, x_N$ correspond to the first $N$ elements of the leading eigenvector of matrix $\bf M$.
\end{lemma}

\begin{proof} Combining (\ref{S2}) and (\ref{B7}), the non-backtracking centrality $x_i$ can be rewritten as
\begin{align}\label{S7}
x_i&=\sum_{j\in\mathcal{N}_i}\frac{1}{\kappa}\sum_{\substack{k\in\mathcal{N}_j\\k\neq i}}v_{j\to k}=\frac{1}{\kappa}\left(\sum_{j\in\mathcal{N}_i}\sum_{k\in\mathcal{N}_j}v_{j\to k}-\sum_{j\in\mathcal{N}_i}v_{j\to i}\right)\nonumber\\
&=\frac{1}{\kappa}\left(\sum_{j\in\mathcal{N}_i}x_j-y_i\right)=\frac{1}{\kappa}\left(\sum_{j=1}^N a_{ij}x_j-\frac{d_i-1}{\kappa}x_i\right).
\end{align}
Recasting (\ref{S7}) in matrix notation, one obtains
\begin{equation}\label{S8}
\left({\bf A}-\frac{1}{\kappa}{\bf D}+\frac{1}{\kappa}{\bf I}\right){\bf x}=\kappa{\bf x},
\end{equation}
where ${\bf x}=(x_1, x_2, \cdots, x_N)^{\top}$ is a vector composed of the non-backtracking centralities of $N$ nodes in $\mathcal{G}$. Equation (\ref{S8}) shows that matrices ${\bf B}$ and ${\bf M}$ have the same set of real eigenvalues. Furthermore,
\begin{equation}\label{Y1}
{\bf M}{\bf z}=\kappa{\bf z}.
\end{equation}
Here ${\bf z}=({\bf x}|\frac{1}{\kappa}{\bf x})$, in which $\bf x$ represents the first $N$ elements of ${\bf z}$ and $\frac{1}{\kappa}{\bf x}$ constitutes the last $N$ elements.  Thus, the leading eigenvalues of matrix ${\bf B}$ and ${\bf M}$ are equal to each other, and the first $N$ elements of {\bf z} correspond to the non-backtracking centralities $x_1$, $x_2$, $\cdots$, $x_N$.
\end{proof}

Lemma~\ref{MandB} indicates that  the computation of the leading eigenvalue $\kappa$ for non-backtracking centrality ${\bf M}$ of order $2E$ and non-backtracking centralities ${\bf x}$ can be reduced to calculating the leading eigenvalue and eigenvector for matrix $\bf M$ of order of $2N$, smaller than the order $2E$ of matrix $\bf B$, especially for dense networks. Thus, we can compute $\kappa$ and $x_i$ very rapidly by evaluating the leading eigenvalue and eigenvector of matrix ${\bf M}$.

\subsection{Definition of Transition Matrix}

According to the bias towards properties of nodes, we can define different biased random walks. Here we propose a novel random walk, non-backtracking centrality random walk (NBCRW), which is a biased one with the transition probability having a bias towards nodes with high non-backtracking centrality.
\begin{definition}
For NBCRW on network $\mathcal{G}$, the element at row $i$ and column $j$ of transition matrix $\bf P$ is
\begin{equation}\label{B8}
p_{ij}=\frac{a_{ij}x_j}{\sum_{k=1}^N a_{ik}x_k}.
\end{equation}
\end{definition}
In other words, the transition probability for NBCRW from node $i$ to its neighbor $j$ is proportional to the non-backtracking centrality of $j$.

In order to investigate the properties of NBCRW on network $\mathcal{G}$, we propose an approach to construct a weighted network $\mathcal{W}$ from the original network $\mathcal{G}$. The weight of each edge in $\mathcal{W}$ is related to the non-backtracking centralities of both nodes connecting the edge in $\mathcal{G}$. We will present that NBCRW on network $\mathcal{G}$ is equivalent to the ordinary random walk~\cite{ZhShCh13} in the corresponding weighted network $\mathcal{W}$, with both random walks having the same properties, such as transition probability, stationary distribution, and hitting times.

\begin{definition}
For an unweighted network $\mathcal{G}(\mathcal{V},\mathcal{E})$, given its adjacency matrix $\bf A$, a diagonal matrix $\bf X$ with its $i$th diagonal entry equal to non-backtracking centrality $x_i$ of node $i$, its corresponding weighted network is defined as $\mathcal{W}(\mathcal{V},\mathcal{E})$, with the weight between nodes $i$ and $j$ given by $w_{ij}=a_{ij}x_i x_j$.
\end{definition}

Let ${\bf W}=(w_{ij})_{N\times N}$ stand for the  adjacency matrix of the weighted network $\mathcal{W}$. Different from the adjacency matrix ${\bf A}$ of binary network $\mathcal{G}$, the elements of ${\bf W}$ are not simply 0 or 1, but are the weights of all pairs of nodes. By definition, we have $\bf W=\bf X\bf A\bf X$. In a weighted network $\mathcal{W}$, the strength~\cite{BaBaPaVe04} of a node $i$ is $s_i=\sum_{k=1}^N w_{ik}=x_i\sum_{k=1}^N a_{ik}x_k$, and  the total strength of the whole network $\mathcal{W}$ is $s=\sum_{i=1}^N\sum_{j=1}^N w_{ij}$. Then, the diagonal strength matrix of $\mathcal{W}$ is defined as ${\bf S}={\rm diag}(s_1, s_2, \cdots, s_N)$, and the Laplacian matrix of $\mathcal{W}$ is defined by ${\bf L}={\bf S}-{\bf W}$.

\begin{theorem}
The transition matrix of NBCRW in an arbitrary connected network $\mathcal{G}$ is identical to the transitional matrix of ordinary random walk in the corresponding weighted network $\mathcal{W}$.
\end{theorem}
\begin{proof} For the ordinary random walk in the weighted network $\mathcal{W}$, the transitional probability from node $i$ to node $j$ is
\begin{equation}
p_{ij}=\frac{w_{ij}}{s_i}=\frac{a_{ij}x_{i}x_{j}}{\sum_{k=1}^N a_{ik}x_{i}x_k}=\frac{a_{ij}x_i}{\sum_{k=1}^N a_{ik}x_k}, \nonumber
\end{equation}
which completely agrees with (\ref{B8}). Therefore, the transition matrix for NBCRW on $\mathcal{G}$ is the same as that of the ordinary random walk on $\mathcal{W}$.
\end{proof}

Since both networks $\mathcal{G}$ and $\mathcal{W}$ have the same topological structure and transition matrix, NBCRW on $\mathcal{G}$ and ordinary random walk on $\mathcal{W}$ also have identical behaviors. In the sequel, we study the properties of NBCRW on $\mathcal{G}$ directly or indirectly by considering those of ordinary random walk on $\mathcal{W}$.

We note that our proposed NBCRW is different from non-backtracking random walk~\cite{AlBeLu07,FiHo13} that is  a random process, during which  the walker never goes back along the edge it just traversed. For a general graph $\mathcal{G}$,   non-backtracking random walk is not a Markov chain on its vertex set, although it can be regarded as a Markov chain on the set of its directed edges~\cite{Ke16}, whose adjacency relation are encoded in non-backtracking matrix $\bf B$. In contrast, NBCRW on the vertex set of $\mathcal{G}$ is a biased Markov chain based on non-backtracking centrality. A main goal of this paper is to unveil the impacts of biases, especially non-backtracking centrality,  on the behaviors of biased random walks.


\subsection{Stationary Distribution}

First, we address the stationary distribution for NBCRW on $\mathcal{G}$.
\begin{theorem}
The stationary distribution for NBCRW on network $\mathcal {G}$ is $\pi^{\rm B}=(\pi_1^{\rm B}, \pi_2^{\rm B}, \cdots, \pi_N^{\rm B})$, where
\begin{equation}\label{B9}
\pi_i^{\rm B}=\left(\frac{\kappa^2-1}{\kappa}+\frac{d_i}{\kappa}\right)\frac{x_i^2}{Q},
\end{equation}
with
\begin{equation}\label{B91}
Q=\sum_{i=1}^N\left(\frac{\kappa^2-1}{\kappa}+\frac{d_i}{\kappa}\right)x_i^2
\end{equation}
being the normalized factor to guarantee $\sum_{i=1}^N\pi_i^{\rm B}=1$.
\end{theorem}

\begin{proof}
First, we prove that $\pi^{\rm B}$ fulfills the detailed balance condition $\pi_i^{\rm B}p_{ij}=\pi_j^{\rm B}p_{ji}$ for different $i$ and $j$. To this end, we require to compute a related quantity $\sum_{k=1}^N a_{ik}x_k$. From (\ref{S8}), we have
\begin{equation}\label{S11}
{\bf A}{\bf x}=\frac{1}{\kappa}{\bf D}{\bf x}+\frac{\kappa^2-1}{\kappa}{\bf x}, \nonumber
\end{equation}
which means
\begin{equation}\label{S12}
\sum_{k=1}^N a_{ik}x_k=\left(\frac{d_i}{\kappa}+\frac{\kappa^2-1}{\kappa}\right)x_i. \nonumber
\end{equation}
Thus, we have
\begin{equation}\label{S13}
\pi_i^{\rm B} p_{ij}=\left(\frac{\kappa^2-1}{\kappa}+\frac{d_i}{\kappa}\right)\frac{x_i^2}{Q} \cdot \frac{a_{ij}x_j}{\sum_{k=1}^N a_{ik}x_k}=\frac{a_{ij} x_i x_j}{Q}. \nonumber
\end{equation}
Similarly, we can get $\pi_j^{\rm B} p_{ji}=\frac{a_{ij} x_i x_j}{Q}$. Hence, the detailed balance condition
\begin{equation}\label{B93}
\pi_i^{\rm B}p_{ij}=\pi_j^{\rm B}p_{ji}
\end{equation}
is satisfied for all pairs of $i$ and $j$.
According to (\ref{B93}), we have
\begin{equation}\label{B92}
\sum_{i=1}^{N}\pi_{i}^{\rm B}p_{ij}=\sum_{i=1}^N\pi_j^{\rm B}p_{ji}=\pi_j^{\rm B}.
\end{equation}
In other words,
\begin{equation}\label{B94}
\pi^{\rm B}{\rm P}=\pi^{\rm B}, \nonumber
\end{equation}
showing that $\pi^{\rm B}$ is the stationary distribution for NBCRW on $\mathcal{G}$.
\end{proof}



\subsection{Hitting Times}


Let $\theta_1, \theta_2, \cdots, \theta_N$ be the $N$ eigenvalues of the Laplacian matrix  $\bf L$ for weighted network $\mathcal{W}$, rearranged as $0=\theta_1<\theta_2\leq\cdots\leq\theta_N$, and let $\phi_{1}, \phi_2, \cdots, \phi_N$ be their corresponding mutually orthogonal eigenvectors of unit length, where $\phi_i=(\phi_{i1}, \phi_{i2}, \cdots, \phi_{iN})^\top$. Then, the hitting times for NBCRW on $\mathcal{G}$ can be expressed in term of the  eigenvalues and eigenvectors of Laplacian matrix of network $\mathcal{W}$.


\begin{theorem}
For non-backtracking centrality based random walk on network $\mathcal{G}$, the hitting time from a node $i$ to another node $j$ is
\begin{equation}\label{Hit1}
T_{ij}=\frac{1}{2}\sum_{z=1}^N s_z\sum_{k=2}^N\frac{1}{\theta_k}\left(\phi_{kj}^2-\phi_{ki}\phi_{kj}-\phi_{kj}\phi_{kz}+\phi_{ki}\phi_{kz}\right),
\end{equation}
the partial mean hitting time to an arbitrary destination node $j$ is
\begin{equation}\label{L1}
T_{j}=\frac{N}{N-1}\sum_{k=2}^{N}\frac{1}{\theta_k}\left(s\times\phi_{kj}^{2}-\phi_{kj}\sum_{z=1}^{N}s_z\phi_{kz}\right),
\end{equation}
and the global mean hitting time for the whole network $\mathcal{G}$ is
\begin{equation}\label{L2}
\langle T\rangle=\frac{s}{N-1}\sum_{k=2}^{N}\frac{1}{\theta_k}.
\end{equation}
\end{theorem}
\begin{proof}
As mentioned earlier, NBCRW on network $\mathcal{G}$ is equivalent to ordinary random walk on its weighted counterpart $\mathcal{W}$. According to our previous result~\cite{LiZh13}, the theorem follows immediately.
\end{proof}

\section{Experiments and results for Model and Realistic Networks}~\label{Sec4}

In this section, we study NBCRW on some classical model networks (e.g. Erd\"os-R\'enyi (ER) network~\cite{ErRe60} and Barab\'asi-Albert (BA) network~\cite{BaAl99}) and real networks, and compare the results of stationary distribution and hitting times for NBCRW with those corresponding to TURW and MERW.

\begin{figure}
\begin{center}
\includegraphics[width=1.15 \linewidth,trim=50 30 0 0]{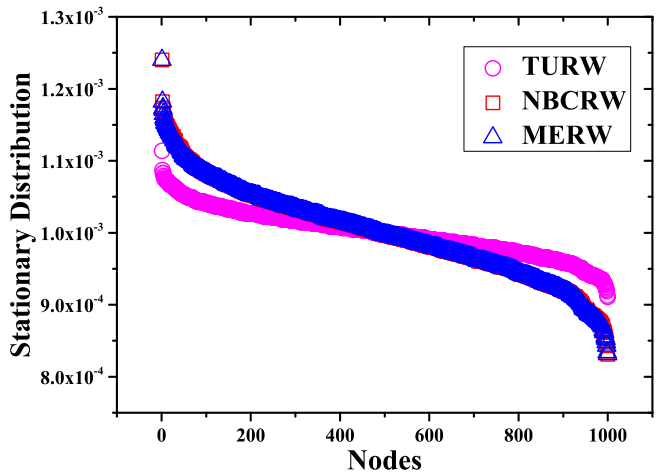}
\end{center}
\caption[kurzform]{Stationary distribution in an ER network with 1000 nodes, where each pair of nodes are connected with probability $p=0.5$. The results for TURW, MERW and NBCRW are obtained by (\ref{B3}), (\ref{B5}) and (\ref{B9}), respectively.According to decreasing order of degree, all the nodes are labeled from 1 to 1000.}\label{ER-St}
\end{figure}

\subsection{Stationary Distribution}


Fig.~\ref{ER-St} shows the stationary distribution for TURW, NBCRW and MERW in an ER network with 1000 nodes. We can see that, for all the three random walks, the stationary probability of a node approximately increases with the degree of the node: for two nodes with different degrees, the stationary probability of the large-degree node is higher than that of the small-degree node. Moreover, the stationary probabilities of the three random walks are all distributed in a narrow range: the largest stationary probability is less than twice of the smallest stationary probability. Thus, there is little difference for the stationary probability of the three random walks. In particular, the stationary distribution of NBCRW and MERW are almost identical to each other. The main reason for this phenomenon is that ER network is homogeneous, with different nodes exhibiting similar structural and dynamical properties.

Fig.~\ref{BA-St} exhibits the behaviors of stationary distribution for TURW, NBCRW and MERW on a BA network with 1000 nodes and average degree $4$. We can see that the stationary distributions are heterogeneous for all the three random walks. According to~\eqref{B3} the stationary distribution of TURW is similar to the degree distribution. Fig.~\ref{BA-St}(a) shows that the stationary probability of TURW lies in the interval $[0.0005, 0.02]$. For NBCRW and MERW, the stationary probability lies, respectively, in the intervals $[10^{-6}, 0.11]$ and $[10^{-7}, 0.16]$, the heterogeneous extent of which is more pronounced than that of TURW. In addition to heterogeneous extent, the stationary distribution of the considered random walks has obvious differences. For TURW, the stationary probability of a node is fully determined by its degree: any two nodes with the identical degree have the same stationary probability. For NBCRW and MERW, two different nodes generally have different stationary probabilities, in spite of their degrees. Thus, the stationary probabilities of NBCRW and MERW can discriminate nodes in the BA networks, including those with identical degree.

However, even for NBCRW and MERW in BA networks, their stationary probabilities differ greatly from each other. For the hub node 1, the stationary probability for MERW is greater than that of NBCRW; while for small-degree nodes, excluding those neighboring nodes of the hub, the stationary probability of a node for MERW is much lower than that corresponding to NBCRW. The insets show that for those 200 small-degree nodes with the lowest stationary probabilities, their stationary probabilities are almost below $10^{-5}$ for MERW, but there are over 150 nodes with stationary probabilities larger than $10^{-5}$ for NBCRW.


\begin{figure*}
\begin{center}
\includegraphics[width=1.02 \linewidth,trim=0 0 0 0]{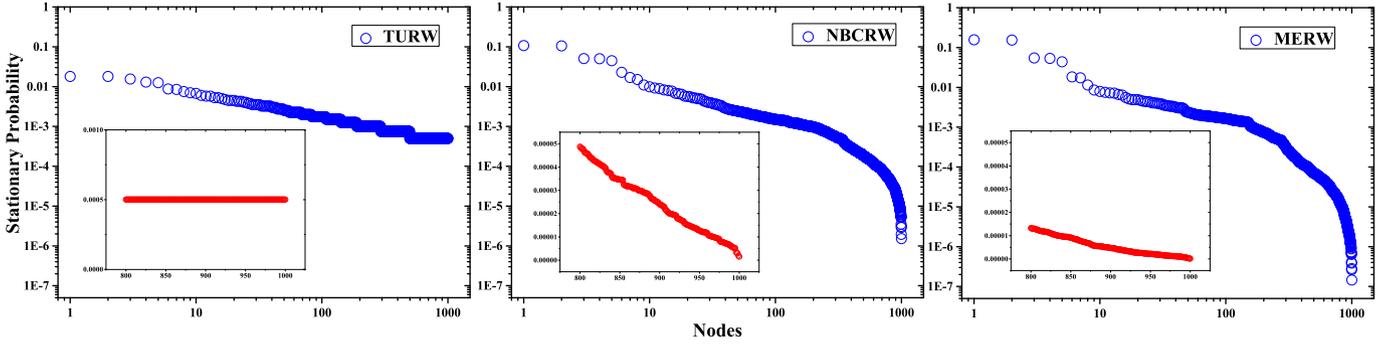}
\end{center}
\caption[kurzform]{Stationary distribution in a BA network with 1000 nodes and average degree $4$. (a): Stationary distribution of TURW, calculated by (\ref{B3}). (b): Stationary distribution of NBCRW, calculated by (\ref{B9}). (c): Stationary distribution of MERW, calculated by (\ref{B5}). According to stationary probability, all the 1000 nodes are labelled from 1 to 1000. The insets are the stationary probabilities of the 200 nodes with the smallest stationary probability.}\label{BA-St}
\end{figure*}


In order to reflect the heterogeneous extent of stationary distributions between NBCRW and MERW in BA networks, we compute the inverse participation ratio $S=\sum_{i=1}^{N}\pi_i^2$, which is a standard quantity characterizing localization or inhomogeneity of an indicator~\cite{BeDe70}: the larger the value $S=\sum_{i=1}^{N}\pi_i^2$, the more heterogeneous the stationary distribution. In Table~\ref{StTab}, we list the inverse participation ratio for NBCRW and MERW on some model and real networks. 
From Table~\ref{StTab} we can see that for all considered model and real networks, the heterogeneity of stationary distribution of MERW is more pronounced than that of NBCRW.


\begin{table}
\caption{Inverse participation ratio of stationary distribution for NBCRW and MERW in a variety of networks. }\label{StTab}
\normalsize
\centering
\begin{tabular}{|c|c|c|c|}
\hline
\raisebox{-0.5ex}{Network} & \raisebox{-0.5ex}{Size} & \raisebox{-0.5ex}{NBCRW} & \raisebox{-0.5ex}{MERW}\\[0.5ex]
\hline
\hline
\raisebox{-0.5ex}{BA-model} & \raisebox{-0.5ex}{$5000$} & \raisebox{-0.5ex}{$2.346\times10^{-2}$} & \raisebox{-0.5ex}{$8.589\times10^{-2}$}\\[0.5ex]
\hline
\raisebox{-0.5ex}{WS-model~\cite{WaSt98}} & \raisebox{-0.5ex}{$1000$} & \raisebox{-0.5ex}{$1.692\times10^{-3}$} & \raisebox{-0.5ex}{$2.510\times10^{-3}$}\\[0.5ex]
\hline
\raisebox{-0.5ex}{Dolphins~\cite{LuScBoHaSlDa03}} & \raisebox{-0.5ex}{$53$} & \raisebox{-0.5ex}{$4.938\times10^{-2}$} & \raisebox{-0.5ex}{$5.312\times10^{-2}$}\\[0.5ex]
\hline
\raisebox{-0.5ex}{ca-NetSci~\cite{Ne06}} & \raisebox{-0.5ex}{$379$} & \raisebox{-0.5ex}{$7.372\times10^{-2}$} & \raisebox{-0.5ex}{$7.938\times10^{-2}$}\\[0.5ex]
\hline
\raisebox{-0.5ex}{C.elegans~\cite{DuAr05}} & \raisebox{-0.5ex}{$448$} & \raisebox{-0.5ex}{$3.369\times10^{-2}$} & \raisebox{-0.5ex}{$4.057\times10^{-2}$}\\[0.5ex]
\hline
\raisebox{-0.5ex}{E-mail~\cite{GuDaDiGiAr03}} & \raisebox{-0.5ex}{$1133$} & \raisebox{-0.5ex}{$8.517\times10^{-3}$} & \raisebox{-0.5ex}{$9.557\times10^{-3}$}\\[0.5ex]
\hline
\raisebox{-0.5ex}{P2P~\cite{LeKlFa07}} & \raisebox{-0.5ex}{$6299$} & \raisebox{-0.5ex}{$7.665\times10^{-3}$} & \raisebox{-0.5ex}{$7.945\times10^{-3}$}\\[0.5ex]
\hline
\end{tabular}
\end{table}


\subsection{Hitting times}

Analogous to the case of stationary distribution, there are little dissimilarity for the behaviors of hitting times between NBCRW, MERW and TURW for homogeneous networks, e.g., ER networks. Below we study hitting times on heterogeneous networks, focusing on two representative cases: mean hitting time to the hub node $T_{\rm H}$ and the global mean hitting time $\langle T\rangle$. Figs.~\ref{BA-Hub} and \ref{BA-GATT} display, respectively, $T_{\rm H}$ and $\langle T\rangle$ for the three random walks in BA networks with node number $N$ changing from $ 1000$ to $10000$.



\begin{figure}
\begin{center}
\includegraphics[width=1.1
\linewidth,trim=50 50 0 50]{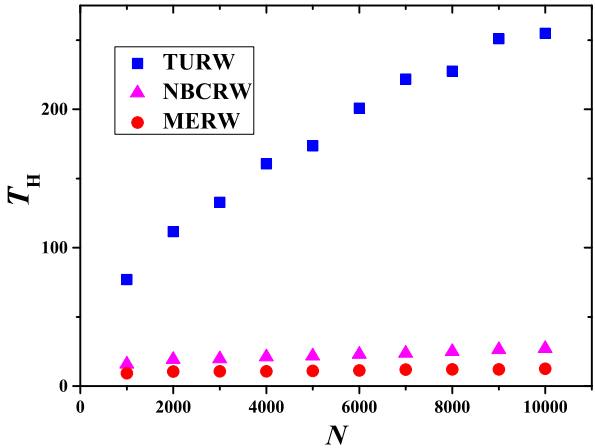}
\end{center}
\caption[kurzform]{Mean hitting time to the hub node for different random walks in BA networks with average degree $4$. The calculation of $T_{\rm H}$ for TURW, MERW and NBCRW are based on (\ref{B32}), (\ref{B52}) and (\ref{L1}), respectively.}\label{BA-Hub} 
\end{figure}

\begin{figure}
\begin{center}
\includegraphics[width=1.1 \linewidth,trim=40 40 0 0]{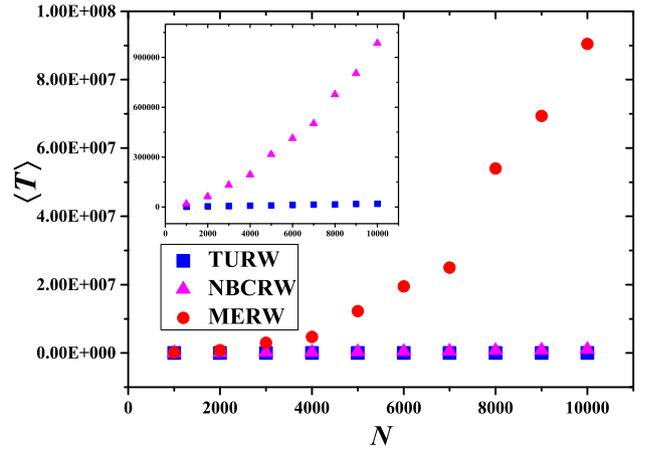}
\end{center}
\caption[kurzform]{Global mean hitting times for TURW, NBCRW and MERW in BA networks with average degree $4$. The inset provides results for TURW and NBCRW for comparison. The calculation of $\langle T\rangle$ for TURW, MERW and NBCRW are based on (\ref{B33}), (\ref{B53}) and (\ref{L2}), respectively.}\label{BA-GATT} 
\end{figure}


From Fig.~\ref{BA-Hub}, we can see that when the hub is the target node, the mean absorbing time is the least for MERW, slightly smaller than that for NBCRW. In contrast, the mean absorbing time to the hub for TURW is significantly higher than those for MERW and NBCRW, which are all in inverse proportion to their corresponding stationary probabilities. In a previous work~\cite{LiZh14}, we have proved that in BA networks, the asymptotical scaling for mean hitting time to the hub for MERW and TURW are $\ln{N}$ and $N^{1/2}$, respectively, both of which are consistent with Fig.~\ref{BA-Hub}.


As opposed to the sublinear scaling of partial mean hitting time $T_H$ to the hub for TURW, NBCRW and MERW in BA networks, the global mean hitting time $\langle T\rangle$ for the three random walks behaves linearly for TURW and superlinearly for NBCRW and MERW, as indicated in Fig.~\ref{BA-GATT}. Although for NBCRW and MERW $\langle T\rangle \sim N^{\rho}$ with $\rho >1$, the power exponent $\rho$ of NBCRW is less than that of MERW. In addition, combining the results in Figs.~\ref{BA-Hub} and \ref{BA-GATT}, we found that among the three random walks, $T_H$ is the largest and $\langle T\rangle$ is the lowest for TURW, with the latter achieving the possible minimal scaling for TURW on all connected networks~\cite{TeBeVo09}; $T_H$ is the smallest and $\langle T\rangle$ is the largest For MERW. for NBCRW, both $T_H$ and $\langle T\rangle$ lie between those associated with TURW and MERW. Thus, for TURW, NBCRW and MERW on a heterogeneous network, the mean absorbing time to a particular target is not representative of the network.

\begin{table*}
\caption{Mean hitting time to a hub node and global mean hitting time for TURW, NBCRW and MERW in a variety of networks.}\label{MFPTTab}
\normalsize
\centering
\begin{tabular}{|c|c|c|c|c|c|c|c|}
\hline
\raisebox{-0.5ex}{Network} & \raisebox{-0.5ex}{Size} & \raisebox{-0.5ex}{$T_{\rm H}^{\rm T}$} & \raisebox{-0.5ex}{$T_{\rm H}^{\rm B}$} & \raisebox{-0.5ex}{$T_{\rm H}^{\rm M}$} & \raisebox{-0.5ex}{$\langle T\rangle^{\rm T}$} & \raisebox{-0.5ex}{$\langle T\rangle^{\rm B}$} & \raisebox{-0.5ex}{$\langle T\rangle^{\rm M}$}\\[0.5ex]
\hline
\hline
\raisebox{-0.5ex}{BA-model} & \raisebox{-0.5ex}{5000} & \raisebox{-0.5ex}{173.9} & \raisebox{-0.5ex}{22.29} & \raisebox{-0.5ex}{9.466} & \raisebox{-0.5ex}{9837} & \raisebox{-0.5ex}{$3.005\times10^5$} & \raisebox{-0.5ex}{$7.960\times10^6$} \\[0.5ex]
\hline
\raisebox{-0.5ex}{WS-model} & \raisebox{-0.5ex}{1000} & \raisebox{-0.5ex}{557.4} & \raisebox{-0.5ex}{246.7} & \raisebox{-0.5ex}{147.0} & \raisebox{-0.5ex}{1667} & \raisebox{-0.5ex}{2440} & \raisebox{-0.5ex}{3396}\\[0.5ex]
\hline
\raisebox{-0.5ex}{Dolphins} & \raisebox{-0.5ex}{53} & \raisebox{-0.5ex}{39.79} & \raisebox{-0.5ex}{14.43} & \raisebox{-0.5ex}{13.05} & \raisebox{-0.5ex}{106.6} & \raisebox{-0.5ex}{3175} & \raisebox{-0.5ex}{6684}\\[0.5ex]
\hline
\raisebox{-0.5ex}{ca-NetSci} & \raisebox{-0.5ex}{379} & \raisebox{-0.5ex}{488.6} & \raisebox{-0.5ex}{14.20} & \raisebox{-0.5ex}{12.56} & \raisebox{-0.5ex}{1892} & \raisebox{-0.5ex}{$2.980\times10^{12}$} & \raisebox{-0.5ex}{$2.767\times10^{13}$}\\[0.5ex]
\hline
\raisebox{-0.5ex}{C.elegans} & \raisebox{-0.5ex}{448} & \raisebox{-0.5ex}{19.82} & \raisebox{-0.5ex}{7.929} & \raisebox{-0.5ex}{6.472} & \raisebox{-0.5ex}{1045} & \raisebox{-0.5ex}{$2.398\times10^6$} & \raisebox{-0.5ex}{$4.421\times10^6$}\\[0.5ex]
\hline
\raisebox{-0.5ex}{E-mail} & \raisebox{-0.5ex}{1133} & \raisebox{-0.5ex}{180.0} & \raisebox{-0.5ex}{27.43} & \raisebox{-0.5ex}{24.46} & \raisebox{-0.5ex}{3713} & \raisebox{-0.5ex}{$7.633\times10^{6}$} & \raisebox{-0.5ex}{$1.118\times10^{7}$}\\[0.5ex]
\hline
\raisebox{-0.5ex}{P2P} & \raisebox{-0.5ex}{6299} & \raisebox{-0.5ex}{476.1} & \raisebox{-0.5ex}{51.03} & \raisebox{-0.5ex}{48.67} & \raisebox{-0.5ex}{$2.094\times10^{4}$} & \raisebox{-0.5ex}{$1.232\times10^{10}$} & \raisebox{-0.5ex}{$2.041\times10^{10}$}\\[0.5ex]
\hline
\end{tabular}
\end{table*}


In addition to the BA networks, we also study partial mean hitting time to the hub and global mean hitting time for TURW, NBCRW and MERW in other synthetic and real networks. In Table~\ref{MFPTTab}, we provide related results for these three random walks, where superscripts ${\rm T}$, ${\rm B}$, and ${\rm M}$ are used to represent the quantities corresponding to TURW, NBCRW and MERW, respectively. From Table~\ref{MFPTTab} we observe that $ T_{\rm H}^{\rm T}>T_{\rm H}^{\rm B}> T_{\rm H}^{\rm M}$ and $\langle T\rangle^{\rm M}>\langle T\rangle^{\rm B}>\langle T\rangle^{\rm T}$ for all studied model and realistic networks.


\section{Analytical Results for NBCBRW on Rose Graphs}~\label{SecS}


In the preceding section, we show that in some model and real networks, the behaviors of NBCRW are strikingly different from those of TURW and MERW. Since many real-life networks are scale-free, analytically unveiling the impact of heterogeneous topology on random walks is important for better understanding its dynamical behaviors and applications. In this section, we study analytically and numerically NBCRW, TURW and MERW on a class of heterogeneous rose graphs~\cite{LiHu13}. For a particular rose graph, we obtain closed-form expressions for stationary distribution and hitting times for these three random walks, and obtain numerical results for general rose graphs, which widely differ from one another. Based on the results, we can discover the impact of topological heterogeneity on NBCRW, TURW and MERW are evidently different.


\subsection{Construction of Rose Graphs}

The rose graphs are a family of deterministic networks, which allow to analytically treat some of their structural and  dynamical properties. Let $\mathcal{R}_m^l$ denote the rose graphs, which are constructed by merging $m$ ($m\geq 2$) $l$-length ($l$ is even) cycles at a central hub node. Here we focus on a specific class of rose graphs, $\mathcal{R}_m^4$ with each petal being 4-length rings, see Fig.~\ref{KGSG}(a). It is easy to derive that in $\mathcal{R}_{m}^4$ the total number of nodes is $N_m=3m+1$, and the total number of edges is $E_m=6m+2$.

For the convenience of description, we partition all the $N_m$ nodes of $\mathcal{R}_m^4$ into three classes: hub node,  peripheral nodes, and  internal nodes. The hub node is the unique node of the largest degree, the peripheral nodes are those $m$ nodes farthermost  from the hub node, while the remaining $2m$ nodes are internal nodes, each of which is linked to the hub node and  a peripheral node. Furthermore, the $3m+1$ nodes can be labelled from $1$ to  $3m+1$ in the following way.  We label by $1$ the hub node. For the nodes in the $i$th ($i=1,2,\cdots,m$) petal, the two internal nodes are labeled as $3(i-1)+2$ and $3(i-1)+3$, while the peripheral node is labeled as $3(i-1)+4$.


\begin{figure*}
\begin{center}
\includegraphics[width=0.8 \linewidth,trim=0 0 0 0]{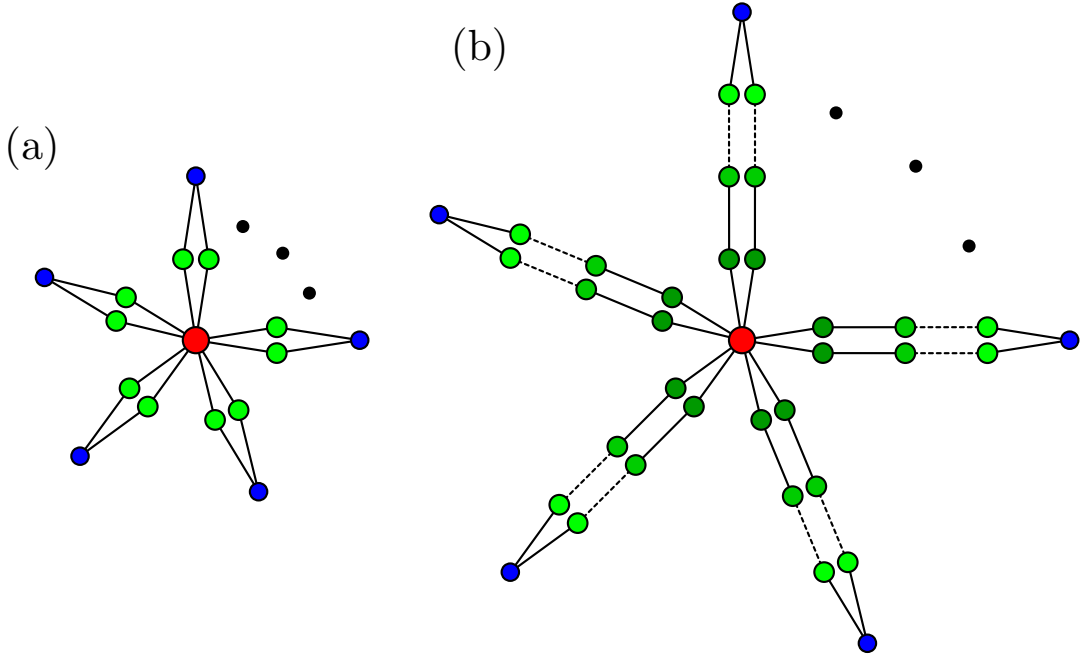}
\end{center}
\caption[kurzform]{Rose graphs. (a): Rose graphs $\mathcal{R}_m^4$. (b): Rose graphs $\mathcal{R}_m^{l}$, where each of $m$ petals is a $l$-length cycle. In either graph, the red node, green nodes and blue nodes stand for the hub, internal nodes, and peripheral nodes, respectively.}\label{KGSG}
\end{figure*}

\subsection{Stationary Distribution} 


For TURW on $\mathcal{R}_m^4$, the stationary distribution can be obtained from~\eqref{B3}. For NBCRW and MERW on  $\mathcal{R}_m^4$, their stationary distributions can also be determined exactly.
\begin{theorem}\label{T1}
For NBCRW on rose graphs $\mathcal{R}_m^4$, the stationary probability at the hub node, an internal node, and a peripheral node are
\begin{align}\label{B14}
\pi_{\rm H}^{\rm B}&=\frac{m}{2(m+\sqrt{2m-1})}=\frac{N_m-1}{2(N_m-1)+2\sqrt{6N_m-15}},
\end{align}
\begin{align}\label{B15}
\pi_{\rm I}^{\rm B}&=\frac{1}{4m}=\frac{3}{4(N_m-1)}
\end{align}
and
\begin{align}\label{B16}
\pi_{\rm P}^{\rm B}&=\frac{m\sqrt{2m-1}-2m+1}{2m(m-1)^2}\nonumber\\
&=\frac{3(N_m-1)\sqrt{6N_m-15}-18N_m+45}{2(N_m-1)(N_m-4)^2},
\end{align}
respectively.
\end{theorem}

The proof is presented in \ref{app-1}.

\begin{theorem}\label{T2}
For MERW on rose graphs $\mathcal{R}_m^4$, the stationary probability at the hub node, an internal node, and a peripheral node are
\begin{equation}\label{B17}
\pi_{\rm H}^{\rm M}=\frac{m}{2m+2}=\frac{N_m-1}{2(N_m-1)+6},
\end{equation}
\begin{equation}\label{B18}
\pi_{\rm I}^{\rm M}=\frac{1}{4m}=\frac{3}{4(N_m-1)},
\end{equation}
and
\begin{equation}\label{B19}
\pi_{\rm P}^{\rm M}=\frac{1}{2m(m+1)}=\frac{9}{2(N_m+2)(N_m-1)},
\end{equation}
respectively.
\end{theorem}

The proof is presented in \ref{app-2}.

Thus far, we have obtained the stationary distribution for NBCRW and MERW on $\mathcal{R}_m^4$. For TURW on $\mathcal{R}_m^4$, the stationary distribution is determined by the degree sequence of nodes and  can be directly computed from~\eqref{B3}, from which we obtain that the stationary probability at the hub node, an interval node, and a peripheral node are $T_{\rm H}^{\rm T}=\frac{m}{3m+1}=\frac{N_m-1}{3N_m}$, $T_{\rm I}^{\rm T}=\frac{1}{3m+1} =\frac{3}{2(N_m-1)}$, and $T_{\rm P}^{\rm T}=\frac{1}{3m+1} =\frac{3}{2(N_m-1)}$, respectively. If we choose stationary probability as an indicator of node importance, the stationary distribution of TURW fails to differentiate the internal nodes and peripheral nodes in $\mathcal{R}_m^4$, since the degree of the two internal and peripheral nodes is equal to each other. We will show that this shortcoming can be overcome by using the stationary distribution for NBCRW and MERW, although they also differ greatly.


For NBCRW and MERW on $\mathcal{R}_m^4$, Theorems \ref{T1} and \ref{T2} show that $\pi_{\rm H}^{\rm M}>\pi_{\rm H}^{\rm B}$, $\pi_{\rm I}^{\rm M}=\pi_{\rm I}^{\rm B}$, and $\pi_{\rm P}^{\rm M}<\pi_{\rm P}^{\rm B}$. Moreover, for both NBCRW and MERW, the stationary probability for internal nodes and peripheral nodes are different, in spite of the fact that their degree is identical. Thus, stationary distribution of NBCRW and MERW can discriminate between an internal node and a peripheral node. However, there exist differences between the stationary probability of NBCRW and MERW. For example, from~(\ref{B16}) and (\ref{B19}) we can see that the stationary probability of a peripheral node for NBCRW gets a fraction $\mathcal{O}(N_m^{-2/3})$, larger than the fraction $\mathcal{O}(N_m^{-2})$ received for MERW. Therefore, in comparison with MERW, the stationary probabilities of NBCRW are distributed over a narrow range of values.



In order to further unveil the distinction of stationary distribution between NBCRW and MERW. We compare the stationary distributions for NBCRW and MERW in the rose graph $\mathcal{R}_3^{20}$ with 58 nodes, among which the hub node has degree 6, while each of the other 57 nodes has a degree of 2. We can classify the 58 nodes in $\mathcal{R}_3^{20}$ by designating a level number to each node according to its shortest distance to the hub node: the hub node is at lever zero, the neighboring nodes of the hub are at level one, and the farthermost nodes from the hub are at level ten.

\begin{figure}
\begin{center}
\includegraphics[width=1.15 \linewidth,trim=50 0 0 0]{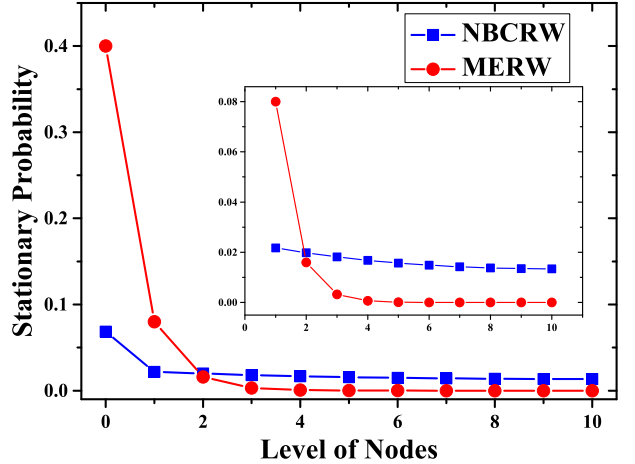}
\end{center}
\caption[kurzform]{Stationary probabilities of different nodes for NBCRW and MERW in the rose graph $\mathcal{R}_3^{20}$.}\label{KGSG-St}
\end{figure}

Fig.~\ref{KGSG-St} provides numerical results of stationary distributions of NBCRW and MERW for every node in  $\mathcal{R}_3^{20}$, which shows that for both NBCRW and MERW, the stationary probability depends on the level: the smaller the level of a node, the larger its stationary probability. However, their differences are also striking. For MERW, the stationary probability almost concentrates around the hub node and its neighbors, with other nodes getting vanishing weight; for NBCRW, although the stationary probability of the hub is also significantly larger than those of other nodes, the stationary probability of every node is nonvanishing and greater than 0.01, as shown in the inset of Fig.~\ref{KGSG-St}. Thus, if we use stationary probability to measure relative importance of nonhub nodes, MERW is hard to distinguish those nodes at large levels, which can be discriminated by NBCRW.



\subsection{Hitting Times}

In addition to stationary distribution, for NBCRW on the rose graph $\mathcal{R}_m^4$, the partial mean hitting time to the hub node and the global mean hitting time of the whole network can also be determined explicitly. For the purpose of comparison, we also provide the corresponding exact results for TURW and MERW. 

\begin{theorem}\label{T3}
The partial mean hitting time to the hub node for TURW, NBCRW and MERW in $\mathcal{R}_m^4$ are
\begin{equation}\label{C4}
T_{\rm H}^{\rm T}=\frac{10}{3},
\end{equation}
\begin{equation}\label{E1}
T_{\rm H}^{\rm B}=\frac{4}{3}+\frac{2\sqrt{2m-1}}{m}=\frac{4}{3}+\frac{2\sqrt{6N_m-15}}{N_m-1}
\end{equation}
and
\begin{equation}\label{E2}
T_{\rm H}^{\rm M}=\frac{4}{3}+\frac{2}{m}=\frac{4}{3}+\frac{6}{N_m-1},
\end{equation}
respectively.
\end{theorem}

The proof is presented in \ref{app-3}

Theorem~\ref{T3} shows that for TURW, NBCRW and MERW in large rose graph $\mathcal{R}_m^4$ ($N_m\to\infty$), the partial mean hitting time to the hub node tends to constants, with $T_{\rm H}^{\rm B}=T_{\rm H}^{\rm M}=4/3$ smaller than $T_{\rm H}^{\rm T}=10/3$. Although for NBCRW and MERW in large rose graphs $\mathcal{R}_m^4$, $T_{\rm H}^{\rm B}$ and $T_{\rm H}^{\rm M}$ are asymptotically equal, their global mean hitting times follow different behaviors, as can be seen from the following theorem.

\begin{theorem}\label{T4}
The global mean hitting times for TURW, NBCRW and MERW in $\mathcal{R}_m^4$ with $N_m=3m+1$ nodes are
\begin{equation}\label{E4}
\langle T\rangle^{\rm T}=\frac{20m(3m-1)}{3(3m+1)}=\frac{20(N_m-1)(N_m-2)}{9N_m},
\end{equation}
\begin{align}\label{E5}
\langle T\rangle^{\rm B}&=\frac{2m^3+12m^2-14m+4}{(3m+1)\sqrt{2m-1}}+\frac{36m^2-8m}{3(3m+1)}\nonumber\\
&=\frac{2N_m^2+30N_m-192}{9\sqrt{6N_m-15}}+\frac{268+20\sqrt{6N_m-15}}{9N_m\sqrt{6N_m-15}}\nonumber\\
&\quad+\frac{12N_m-32}{9}
\end{align}
and
\begin{align}\label{E6}
\langle T\rangle^{\rm M}&=\frac{6m^3+36m^2+10m-12}{9m+3}\nonumber\\
&=\frac{2N_m^3+30N_m^2-36N_m-104}{27N_m},
\end{align}
respectively.
\end{theorem}

The proof is presented in \ref{app-4}.

Theorem \ref{T3} provides succinct dependence relations of $\langle T\rangle^{\rm T}$, $\langle T\rangle^{\rm B}$ and $\langle T\rangle^{\rm M}$ on the network size $N_m$, from which we can find that for large networks (i.e. $N_m\to\infty$), the leading terms for global mean hitting times for TURW, NBERW, and MERW are $\langle T\rangle^{\rm T}\sim N_m$, $\langle T\rangle^{\rm B}\sim \left(N_m\right)^{\frac{3}{2}}$ and $\langle T\rangle^{\rm M}\sim \left(N_m\right)^{2}$, respectively. Thus, $\langle T\rangle^{\rm T}$, $\langle T\rangle^{\rm B}$ and $\langle T\rangle^{\rm M}$ behave differently for the three random walks on $\mathcal{R}_m^4$. For TURW, $\langle T\rangle^{\rm T}$ grows linearly with $N_m$; while for both NBCRW and MERW, $\langle T\rangle^{\rm B}$ and $\langle T\rangle^{\rm M}$ increase superlinearly with $N_m$, with
$\langle T\rangle^{\rm B}$ smaller than $\langle T\rangle^{\rm M}$. These results indicate that, when searching a node distributed uniformly in $\mathcal{R}_m^4$, TURW is the most efficient, while MERW is the most inefficient, as observed in the model and real networks studied in the previous section.



\section{Conclusion}~\label{Sec5}

The application effects of random walks are determined to a large extent by the properties and behaviors of stationary distribution and hitting times. Recent works indicate that biased random walks 
perform better in multiple applications than TURW. Thus, designing appropriate biased random walks and understanding their properties are of significant  importance. In this paper, we defined a new biased random walk, NBCRW, with the bias dependent on the non-backtracking centrality, which is a recently proposed node centrality measure having several advantages over traditional eigenvector centrality metric. We established a theoretical framework for computing quickly the transition probabilities, stationary distribution, and hitting times of NBCRW on a general network.



Within our proposed framework, we studied numerically or analytically NBCRW on some model and realistic networks, and compared the results about stationary distribution and hitting times with those corresponding to TURW and MERW, the latter of which is actually a biased random walk towards selecting neighboring having high eigenvector centrality. We found that for homogeneous networks, the behaviors for stationary distribution and hitting times of the three random walks resemble to each other. However, for heterogeneous networks, there is a big difference in the behaviors of the three random walks. For example, the stationary distribution of NBCRW outperforms TURW and MERW in discriminating nodes, in particular those with identical degree. With respect to hitting times, a walker finds the hub node most quickly when performing MERW, and detects a uniformly selected target most rapidly when executing TURW. For both cases, the hitting times of NBCRW interpolates between TURW and MERW.

In view of the distinctive behaviors of NBCRW, in future it is interesting to explore the applications of NBCRW in different fields, such as community detection, visual saliency, link prediction, and so on.


\section*{Acknowledgements}
This work was supported by the National Natural Science Foundation of China under Grants No. 11275049.

\appendix

\section{Proof of Theorem \ref{T1}}\label{app-1}

	Let $\kappa_1$ denote the leading eigenvalue of matrix ${\bf M}$ corresponding to $\mathcal{R}_m^4$. 
	Let $M_m(\kappa)$ be the characteristic polynomial of matrix ${\bf M}$, i.e.,
	\begin{equation}\label{S15}
	M_m(\kappa)=\det{(\kappa{\bf I}-{\bf M})}.
	\end{equation}
	Then, $\kappa_1$ is the largest root of equation $M_m(\kappa)=0$. Equation~(\ref{S15}) can be recast as
	\begin{align}\label{S16}
	M_m(\kappa)&=\det \left(\begin{array}{cc}\kappa{\bf I}-{\bf A} & {\bf D}-{\bf I}\\
	-{\bf I} & \kappa{\bf I}\end{array}\right)\nonumber\\
	&=\frac{1}{\kappa^{N_m}}\cdot \det \left(\begin{array}{cc} \kappa^2{\bf I}-\kappa{\bf A}-{\bf I}+{\bf D} & {\bf I}-{\bf D}\\
	{\bf O} & \kappa{\bf I}\end{array}\right)\nonumber\\
	&=\det \left((\kappa^2-1){\bf I}+{\bf D}-\kappa{\bf A}\right)\,,
	\end{align}
	which reduces the computation of $M_m(\kappa)$ to computing a determinant of a new matrix ${\bf P}_m=(\kappa^2-1){\bf I}+{\bf D}-\kappa{\bf A}$ of low order. According to the construction of $\mathcal{R}_m^4$, $\det({\bf P}_m)$ can be rephrased in the following form
	\begin{equation}\label{S17}
	\det({\bf P}_m)=\det \left(\begin{array}{ccccc} \kappa^2+2m-1 & -\kappa{\bf e} & -\kappa{\bf e} & \cdots & -\kappa{\bf e}\\
	-\kappa{\bf e}^{\top} & {\bf Q} & {\bf O} & \cdots & {\bf O}\\
	-\kappa{\bf e}^{\top} & {\bf O} & {\bf Q} & \cdots & {\bf O}\\
	\vdots & \vdots & \vdots & \quad & \vdots\\
	-\kappa{\bf e}^{\top} & {\bf O} & {\bf O} & \cdots & {\bf Q}\end{array}\right),
	\end{equation}
	where ${\bf e}=(1, 1, 0)$, $\bf O$ is the $3 \times 3$ zero matrix, and $\bf Q$ is a $3 \times 3$ matrix given by
	\begin{equation}\label{S18}
	{\bf Q}=\left(\begin{array}{ccc} \kappa^2+1 & 0 & -\kappa \\
	0 & \kappa^2+1 & -\kappa\\
	-\kappa & -\kappa & \kappa^2+1\end{array}\right).
	\end{equation}
	Note that the first row of matrix ${\bf P}_m$ on the right-hand side (rhs) of (\ref{S17}) can be regarded as the sum (i.e., a linear combination with all scalars being 1) of the following $m+1$ vectors: $(\kappa^2+2m-1, {\bf 0}, {\bf 0}, \cdots, {\bf 0})$, $(0, -\kappa{\bf e}, {\bf 0}, \cdots, {\bf 0})$, $(0, {\bf 0}, -\kappa{\bf e}, \cdots, {\bf 0})$, $\cdots$, $(0, {\bf 0}, {\bf 0}, \cdots, -\kappa{\bf e})$, where ${\bf 0}$ represents the zero vector $(0, 0, 0)$. According to the properties of determinants, 
	$\det({\bf P}_m)$ can be rewritten as
	\begin{align}\label{S19}
	\det({\bf P}_m)&=\det \left(\begin{array}{ccccc}\kappa^2+2m-1 & {\bf 0} & {\bf 0} & \cdots & {\bf 0} \\
	-\kappa{\bf e}^{\top} & {\bf Q} & {\bf O} & \cdots & {\bf O} \\
	-\kappa{\bf e}^{\top} & {\bf O} & {\bf Q} & \cdots & {\bf O} \\
	\vdots & \vdots & \vdots & \quad & \vdots \\
	-\kappa{\bf e}^{\top} & {\bf O} & {\bf O} & \cdots & {\bf Q}\end{array}\right )\nonumber\\
	&\quad+m\cdot\det \left(\begin{array}{ccccc}0 & -\kappa{\bf e} & {\bf 0} & \cdots & {\bf 0} \\
	-\kappa{\bf e}^{\top} & {\bf Q} & {\bf O} & \cdots & {\bf O} \\
	-\kappa{\bf e}^{\top} & {\bf O} & {\bf Q} & \cdots & {\bf O} \\
	\vdots & \vdots & \vdots & \quad & \vdots \\
	-\kappa{\bf e}^{\top} & {\bf O} & {\bf O} & \cdots & {\bf Q}\end{array}\right)\nonumber\\
	&=(\kappa^2+2m-1)(\det{\bf Q})^m+\nonumber\\
	&\quad m(\det{\bf Q})^{m-1}\det \left(\begin{array}{cc}0 & -\kappa{\bf e}\\ -\kappa{\bf e}^{\top} & {\bf Q}\end{array}\right).
	\end{align}
	Based on \eqref{S18}, we obtain
	\begin{equation}\label{S20}
	\det{(\bf Q)}=\kappa^6+\kappa^4+\kappa^2+1
	\end{equation}
	and
	\begin{equation}\label{S21}
	\det \left(\begin{array}{cc}0 & -\kappa{\bf e}\\ -\kappa{\bf e}^{\top} & {\bf Q}\end{array}\right)=-2\kappa^6-4\kappa^4-2\kappa^2.
	\end{equation}
	Inserting \eqref{S20} and \eqref{S21} into \eqref{S19} yields
	\begin{align}\label{S22}
	M_m(\kappa)&=\det({\bf P}_m)=(\kappa+1)(\kappa-1)(\kappa^2+1)\nonumber\\
	&\quad(\kappa^4-2m+1)(\kappa^6+\kappa^4+\kappa^2+1)^{m-1}.
	\end{align}
	Thus, the largest eigenvalue $\kappa_1$ of matrix ${\bf M}$ is
	\begin{equation}\label{S23}
	\kappa_1=(2m-1)^{\frac{1}{4}}.
	\end{equation}
	
	Next, we continue to derive the eigenvector of unit length corresponding to eigenvalue $\kappa_1$. Let $x_{\rm H}$, $x_{\rm I}$ and $x_{\rm P}$ represent separately the non-backtracking centrality for the hub node, an internal node  and a peripheral node in graph $\mathcal{R}_m^4$.  
	According to \eqref{S8}, $x_{\rm H}$, $x_{\rm I}$ and $x_{\rm P}$ satisfy the following system of equations:
	\begin{equation}\label{S24}
	\left\{
	\begin{array}{c}
	2m x_{\rm I}-(2m-1)\frac{x_{\rm H}}{\kappa}=\kappa x_{\rm H},\\
	x_{\rm H}+x_{\rm P}-\frac{x_{\rm I}}{\kappa}=\kappa x_{\rm I},\\
	x_{\rm H}^2+\left(\frac{x_{\rm H}}{\kappa}\right)^2+2m\left[x_{\rm I}^2+\left(\frac{x_{\rm I}}{\kappa}\right)^2\right]+m\left[x_{\rm P}^2+\left(\frac{x_{\rm P}}{\kappa}\right)^2\right]=1,\nonumber\\
	\end{array}
	\right.
	\end{equation}
	which can be resolved to yield
	\begin{equation}\label{S25}
	\left\{
	\begin{array}{c}
	x_{\rm H}=\frac{2\sqrt{m}\kappa^3}{\sqrt{(\kappa^2+1)(\kappa^8+2\kappa^6+2(8m-3)\kappa^4+2(2m-1)^2\kappa^2+(2m-1)^2)}},\\
	x_{\rm I}=\frac{\kappa^2(\kappa^2+2m-1)}{\sqrt{m(\kappa^2+1)(\kappa^8+2\kappa^6+2(8m-3)\kappa^4+2(2m-1)^2\kappa^2+(2m-1)^2)}},\\
	x_{\rm P}=\frac{\kappa(\kappa^4+2m-1)}{\sqrt{m(\kappa^2+1)(\kappa^8+2\kappa^6+2(8m-3)\kappa^4+2(2m-1)^2\kappa^2+(2m-1)^2)}}.\\
	\end{array}
	\right.
	\end{equation}
	
	Then the normalized factor $Q$ can be computed by
	\begin{align}\label{S26}
	Q&=\sum_{i=1}^{N_m}\left(\frac{\kappa^2-1}{\kappa}+\frac{d_i}{\kappa}\right)x_i^2\nonumber\\
	&=\left(\frac{\kappa^2-1}{\kappa}+\frac{2m}{\kappa}\right)x_{\rm H}^2+2m\left(\frac{\kappa^2-1}{\kappa}+\frac{2}{\kappa}\right)x_{\rm I}\nonumber\\
	&\quad+m\left(\frac{\kappa^2-1}{\kappa}+\frac{2}{\kappa}\right)x_{\rm P}\nonumber\\
	&=\frac{2(2m-1)^{\frac{1}{4}}[2m^2+m-1+\sqrt{2m-1}(3m-1)]}{(\sqrt{2m-1}+1)[m(5+\sqrt{2m-1})-2]}.
	\end{align}
	Plugging (\ref{S23}), (\ref{S25}) and (\ref{S26}) into (\ref{B9}) and considering the relation $N_m=3m+1$, the theorem follows. 

\section{Proof of Theorem \ref{T2}}\label{app-2}

	Let $\lambda_1$ be the leading eigenvalue of adjacency matrix ${\bf A}$ for graph $\mathcal{R}_m^4$. And let $\mu_{\rm H}$, $\mu_{\rm I}$ and $\mu_{\rm P}$ be the elements of the leading eigenvector of unit length corresponding to the hub node, an internal node and a peripheral node, respectively. Then, 
	\begin{equation}\label{S31}
	\left\{
	\begin{array}{ccc}
	2m\mu_{\rm I}&=&\lambda_1\mu_{\rm H},\\
	\mu_{\rm H}+\mu_{\rm P}&=&\lambda_1\mu_{\rm I},\\
	2\mu_{\rm I}&=&\lambda_1\mu_{\rm P}.\\
	\end{array}
	\right.
	\end{equation}
	Doing some simple algebra operations on (\ref{S31}), we have
	\begin{equation}\label{S32}
	\left\{
	\begin{array}{ccc}
	\frac{2m\mu_{\rm I}}{\mu_{\rm H}+\mu_{\rm P}}=\frac{\mu_{\rm H}}{\mu_{\rm I}},\\
	\frac{\mu_{\rm H}+\mu_{\rm P}}{2\mu_{\rm I}}=\frac{\mu_{\rm I}}{\mu_{\rm P}},\\
	\end{array}
	\right.
	\end{equation}
	which, together with the normalization condition $\mu_{\rm H}^2+2m\mu_{\rm I}^2+2\mu_{\rm P}^2=1$, is solved to yield
	\begin{equation}\label{S33}
	\left\{
	\begin{array}{ccc}
	\mu_{\rm H}&=&\frac{\sqrt{m}}{\sqrt{2m+2}},\\
	\mu_{\rm I}&=&\frac{1}{2\sqrt{m}},\\
	\mu_{\rm P}&=&\frac{1}{\sqrt{m}\sqrt{2m+2}}.\\
	\end{array}
	\right.
	\end{equation}
	Combining \eqref{S33} and the relation $N_m=3m+1$, the theorem follows from \eqref{B5}. 


\section{Proof of Theorem \ref{T3}}\label{app-3}

	For TURW (MERW, NBCRW) in $\mathcal{R}_m^4$, let $T_{\rm I\to \rm H}^{\rm T}$ ($T_{\rm I\to \rm H}^{\rm M}$, $T_{\rm I\to \rm H}^{\rm B}$) be the hitting time from an internal node to the hub and, and let $T_{\rm P\to \rm H}^{\rm T}$ ($T_{\rm P\to \rm H}^{\rm M}$, $T_{\rm P\to \rm H}^{\rm B}$) be the hitting time from a peripheral node to the hub. Then, by definition, for the three random walks in $\mathcal{R}_m^4$ the partial mean hitting time to the hub can be computed in a uniform formula as
	\begin{equation}\label{S37}
	T_{\rm H}^{\rm Z}=\frac{1}{3}(2T_{\rm I\to\rm H}^{\rm Z}+T_{\rm P\to\rm H}^{\rm Z}),
	\end{equation}
	where ${\rm Z}$ can be ${\rm T}$, ${\rm M}$, or ${\rm B}$. We next determine the partial mean hitting time to the hub for the three considered random walks.
	
	\emph{Case I: TURW}. Since TURW is unbiased, the quantities $T_{\rm I\to\rm H}^{\rm T}$ and $T_{\rm P\to\rm H}^{\rm T}$ satisfy the following relations:
	\begin{equation}\label{S38}
	T_{\rm I\to\rm H}^{\rm T}=\frac{1}{2}+\frac{1}{2}(1+T_{\rm P\to\rm H}^{\rm T}) \nonumber
	\end{equation}
	and
	\begin{equation}\label{S39}
	T_{\rm P\to\rm H}^{\rm T}=1+T_{\rm I\to\rm H}^{\rm T}, \nonumber
	\end{equation}
	from which we have
	\begin{equation}\label{S40}
	T_{\rm I\to\rm H}^{\rm T}=3
	\end{equation}
	and
	\begin{equation}\label{S41}
	T_{\rm P\to\rm H}^{\rm T}=4.
	\end{equation}
	Plugging (\ref{S40}) and (\ref{S41}) into (\ref{S37}) yields 
	\begin{equation}\label{S42}
	T_{\rm H}^{\rm T}=\frac{10}{3}.\nonumber
	\end{equation}
	
	\emph{Case II: NBCRW}. 
	We first determine the transition probabilities between different nodes for NBCRW in $\mathcal{R}_m^4$. If the walker is currently at a peripheral node, at next time step it will jump to either of the two internal nodes adjacent to it; 
	if the current location of the walker is an internal node, according to \eqref{B8} and \eqref{S25}, at next time step, the probability of the walker at the hub node or a peripheral node is $\frac{x_{\rm H}}{x_{\rm H}+x_{\rm P}}=\frac{m}{m+\sqrt{2m-1}}$ and  $\frac{x_{\rm P}}{x_{\rm H}+x_{\rm P}}=\frac{\sqrt{2m-1}}{m+\sqrt{2m-1}}$, respectively. Then, we can establish the relations between $T_{\rm I\to\rm H}^{\rm B}$ and $T_{\rm P\to\rm H}^{\rm B}$ as
	\begin{equation}\label{S43}
	T_{\rm I\to\rm H}^{\rm B}=\frac{m}{m+\sqrt{2m-1}}+\frac{\sqrt{2m-1}}{m+\sqrt{2m-1}}(1+T_{\rm P\to\rm H}^{\rm B})\nonumber
	\end{equation}
	and
	\begin{equation}\label{S44}
	T_{\rm P\to\rm H}^{\rm B}=1+T_{\rm I\to\rm H}^{\rm B},\nonumber
	\end{equation}
	which can be solved to yield
	\begin{equation}\label{S45}
	T_{\rm I\to\rm H}^{\rm B}=1+\frac{2\sqrt{2m-1}}{m}
	\end{equation}
	and
	\begin{equation}\label{S46}
	T_{\rm P\to\rm H}^{\rm B}=2+\frac{2\sqrt{2m-1}}{m}.\nonumber
	\end{equation}
	Thus, the mean hitting time to the hub for NBCRW in $\mathcal{R}_m^4$ is
	\begin{align}\label{S47}
	T_{\rm H}^{\rm B}&=\frac{1}{3}(2T_{\rm I\to\rm H}^{\rm B}+T_{\rm P\to\rm H}^{\rm B})\nonumber \\&=\frac{4}{3}+\frac{2\sqrt{2m-1}}{m}=\frac{4}{3}+\frac{2\sqrt{6N_m-15}}{N_m-1}.\nonumber
	\end{align}
	
	\emph{Case III: MERW}. 
	For MERW in $\mathcal{R}_m^4$, the transition probability from a node $i$ to one of its neighbor $j$ is $\mu_j/\sum_{k}a_{ik}\mu_k$. Then, according to (\ref{S33}), we obtain that the probabilities from an internal node to the hub node and the peripheral neighbor node are $\frac{\mu_{\rm H}}{\mu_{\rm H}+\mu_{\rm P}}=\frac{m}{m+1}$ and $\frac{\mu_{\rm P}}{\mu_{\rm H}+\mu_{\rm P}}=\frac{1}{m+1}$, respectively. Thus, we can establish the following relations for $T_{\rm I\to\rm H}^{\rm M}$ and $T_{\rm P\to \rm H}^{\rm M}$:
	\begin{equation}\label{S48}
	T_{\rm I\to\rm H}^{\rm M}=\frac{m}{m+1}+\frac{1}{m+1}(1+T_{\rm P\to \rm H}^{\rm M})
	\end{equation}
	and
	\begin{equation}\label{S49}
	T_{\rm P\to \rm H}^{\rm M}=1+T_{\rm I\to \rm H}^{\rm M}.
	\end{equation}
	Resolving (\ref{S48}) and (\ref{S49}), one obtains the analytical expressions for $T_{\rm I\to \rm H}^{\rm M}$ and $T_{\rm P\to \rm H}^{\rm M}$ as
	\begin{equation}\label{S50}
	T_{\rm I\to \rm H}^{\rm M}=\frac{m+2}{m}
	\end{equation}
	and
	\begin{equation}\label{S51}
	T_{\rm P\to \rm H}^{\rm M}=\frac{2(m+1)}{m}.
	\end{equation}
	According \eqref{S37}, (\ref{S50}) and (\ref{S51}), the mean hitting time to the hub for MERW in $\mathcal{R}_m^4$ is obtained to be
	\begin{equation}\label{S52}
	T_{\rm H}^{\rm M}=\frac{1}{3}(2T_{\rm I\to \rm H}^{\rm M}+T_{\rm P\to \rm H}^{\rm M})=\frac{4}{3}+\frac{2}{m}=\frac{4}{3}+\frac{6}{N_m-1}.\nonumber
	\end{equation}
	This completes the proof.

\section{Proof of Theorem \ref{T4}}\label{app-4}

	We use superscript ${\rm Z} \in \{ {\rm T}, {\rm M}, {\rm B} \}$ to differentiate related quantities for TURW, MERW, NBCRW on $\mathcal{R}_m^4$.  For example, $\langle T\rangle^{{\rm T}}$ $\left(\langle T\rangle^{{\rm M}}, \langle T\rangle^{{\rm B}}\right)$ presents global mean hitting time for TURW (MERW, NBCRW) on $\mathcal{R}_m^4$. By definition,
	\begin{equation}\label{S53}
	\langle T\rangle^{{\rm Z}}=\frac{1}{N_m(N_m-1)}\sum_{i=1}^{N_m}\sum_{j=1}^{N_m}T_{i\to j}^{{\rm Z}}=\frac{T_{\rm tot}^{\rm Z}}{N_m(N_m-1)},
	\end{equation}
	where $T_{i \to j}^{\rm Z}$ is the hitting time from node $i$ to node $j$ in $\mathcal{R}_m^4$,  and
	\begin{equation}\label{S54}
	T_{\rm tot}^{\rm Z}=\sum_{i=1}^{N_m}\sum_{j=1}^{N_m}T_{i \to j}^{\rm Z}\nonumber
	\end{equation}
	denotes the sum of  hitting times over all $N_m(N_m-1)$ node pairs in $\mathcal{R}_m^4$.
	
	Thus, in order to obtain $\langle T\rangle^{\rm Z}$, we only need to determine $T_{\rm tot}^{\rm Z}$. By construction, 
	$T_{\rm tot}^{\rm Z}$ can be decomposed into two terms as
	\begin{equation}\label{S56}
	T_{\rm tot}^{\rm Z}=m T_{\rm tot, 1 }^{\rm Z}+m(m-1)T_{\rm tot, 2 }^{\rm Z},
	\end{equation}
	where $T_{\rm tot, 1 }^{\rm Z}$ is the sum of hitting times between all pairs of nodes belonging to one of the $m$ petals, and $T_{\rm tot, 2 }^{\rm Z}$ is the sum of hitting times between  all pairs of nodes in different petals. We now compute $T_{\rm tot, 1 }^{\rm Z}$ and $T_{\rm tot, 2 }^{\rm Z}$.
	
	For $T_{\rm tot, 1 }^{\rm Z}$, it can be evaluated by
	\begin{align}\label{S57}
	T_{\rm tot, 1 }^{\rm Z}&=2T_{\rm I\to\rm H}^{\rm Z}+T_{\rm P\to\rm H}^{\rm Z}+2(T_{\rm H\to\rm I}^{\rm Z}+T_{\rm I\to\rm I}^{\rm Z}+T_{\rm P\to\rm I}^{\rm Z})\nonumber\\
	&\quad+T_{\rm H\to\rm P}^{\rm Z}+2T_{\rm I\to\rm P}^{\rm Z},
	\end{align}
	where $T_{\rm X\to\rm Y}^{\rm Z}$ represents the hitting time from a node in class $X$ to another node in class $Y$ with both nodes belonging to the same petal. For instance, $T_{\rm I\to\rm P}^{\rm Z}$ is the hitting time from an internal node to the  peripheral node in the same petal, and $T_{\rm I\to\rm I}^{\rm Z}$ is the hitting time from one internal node to the other internal node in the same petal. 
	For $T_{\rm tot, 2}^{\rm Z}$, we have
	\begin{align}\label{S58}
	T_{\rm tot, 2}^{\rm Z}&=2(3T_{\rm I\to\rm H}^{\rm Z}+2T_{\rm H\to\rm I}^{\rm Z}+T_{\rm H\to\rm P}^{\rm Z})\nonumber\\
	&\quad+(3T_{\rm P\to\rm H}^{\rm Z}+2T_{\rm H\to\rm I}^{\rm Z}+T_{\rm H\to\rm P}^{\rm Z}).
	\end{align}
	Plugging (\ref{S57}) and (\ref{S58}) into (\ref{S56}) yields
	\begin{align}\label{S59}
	T_{\rm tot}^{\rm Z}&=(6m^2-4m)T_{\rm I\to\rm H}^{\rm Z}+(3m^2-2m)T_{\rm P\to\rm H}^{\rm Z}\nonumber\\
	&\quad+(6m^2-4m)T_{\rm H\to\rm I}^{\rm Z}+2mT_{\rm I\to\rm I}^{\rm Z}+2mT_{\rm P\to\rm I}^{\rm Z}\nonumber\\
	&\quad+(3m^2-2m)T_{\rm H\to\rm P}^{\rm Z}+2mT_{\rm I\to\rm P}^{\rm Z}.
	\end{align}
	We are now ready to determine the global mean hitting time for TURW, NBCRW and MERW in $\mathcal{R}_m^4$ by evaluating those quantities on the rhs of (\ref{S59}).
	
	\emph{Case I: TURW.} Since the quantities $T_{\rm I\to\rm H}^{\rm T}$ and $T_{\rm P\to\rm H}^{\rm T}$  have been obtained earlier, we only require to determine $T_{\rm H\to\rm I}^{\rm T}$, $T_{\rm I\to\rm I}^{\rm T}$, $T_{\rm P\to\rm I}^{\rm T}$, $T_{\rm H\to\rm P}^{\rm T}$ and $T_{\rm I\to\rm P}^{\rm T}$, which obey the following relations:
	\begin{equation}\label{S60}
	T_{\rm H\to \rm I}^{\rm T}=\frac{m-1}{ m}(1+T_{\rm I\to\rm H}^{\rm T}+T_{\rm H\to\rm I}^{\rm T})+\frac{1}{2m}(1+T_{\rm I\to\rm I}^{\rm T})+\frac{1}{2m}, \nonumber
	\end{equation}
	\begin{equation}\label{S61}
	T_{\rm I\to\rm I}^{\rm T}=\frac{1}{2}(1+T_{\rm H\to\rm I}^{\rm T})+\frac{1}{2}(1+T_{\rm P\to\rm I}^{\rm T}), \nonumber
	\end{equation}
	\begin{equation}\label{S62}
	T_{\rm P\to\rm I}^{\rm T}=\frac{1}{2}(1+T_{\rm I\to\rm I}^{\rm T})+\frac{1}{2},\nonumber
	\end{equation}
	\begin{equation}\label{S63}
	T_{\rm H\to\rm P}^{\rm T}=\frac{m-1}{m}(1+T_{\rm I\to \rm H}^{\rm T}+T_{\rm H\to\rm P}^{\rm T})+\frac{1}{m}(1+T_{\rm I\to\rm P}^{\rm T}\nonumber
	\end{equation}
	and
	\begin{equation}\label{S64}
	T_{\rm I\to\rm P}^{\rm T}=\frac{1}{2}(1+T_{\rm H\to\rm P}^{\rm T})+\frac{1}{2}.\nonumber
	\end{equation}
	Using (\ref{S40}), the above equations are solved to obtain
	\begin{equation}\label{S65}
	T_{\rm H\to\rm I}^{\rm T}=6m-3,
	\end{equation}
	\begin{equation}\label{S66}
	T_{\rm I\to\rm I}^{\rm T}=4m,
	\end{equation}
	\begin{equation}\label{S67}
	T_{\rm P\to\rm I}^{\rm T}=2m+1,
	\end{equation}
	\begin{equation}\label{S68}
	T_{\rm H\to\rm P}^{\rm T}=8m-4
	\end{equation}
	and
	\begin{equation}\label{S69}
	T_{\rm I\to\rm P}^{\rm T}=4m-1.
	\end{equation}
	Plugging (\ref{S40})-(\ref{S41}) and (\ref{S65})-(\ref{S69}) into (\ref{S59}) and \eqref{S53}, we obtain the explicit expression for the global mean hitting time $\langle T\rangle^{\rm T}$ for TURW in $\mathcal{R}_m^4$ and its relation between node number $N_m=3m+1$, as given by \eqref{E4}.
	
	\emph{Case II: NBCRW}. 
	For NBCRW on $\mathcal{R}_m^4$, we can establish the following recursive relations for the quantities $T_{\rm H\to\rm I}^{\rm B}$, $T_{\rm I\to\rm I}^{\rm B}$, $T_{\rm P\to\rm I}^{\rm B}$, $T_{\rm H\to\rm P}^{\rm B}$ and $T_{\rm I\to\rm P}^{\rm B}$:
	\begin{equation}\label{S71}
	T_{\rm H\to\rm I}^{\rm B}=\frac{m-1}{ m}(1+T_{\rm I\to\rm H}^{\rm B}+T_{\rm H\to\rm I}^{\rm B})+\frac{1}{2m}(1+T_{\rm I\to\rm I}^{\rm B})+\frac{1}{2m},\nonumber
	\end{equation}
	\begin{equation}\label{S72}
	T_{\rm I\to\rm I}^{\rm B}=\frac{x_{\rm H}}{x_{\rm H}+x_{\rm P}}(1+T_{\rm H\to\rm I}^{\rm B})+\frac{x_{\rm P}}{x_{\rm H}+x_{\rm P}}(1+T_{\rm P\to\rm I}^{\rm B}),\nonumber
	\end{equation}
	\begin{equation}\label{S73}
	T_{\rm P\to\rm I}^{\rm B}=\frac{1}{2}(1+T_{\rm I\to\rm I}^{\rm B})+\frac{1}{2},\nonumber
	\end{equation}
	\begin{equation}\label{S74}
	T_{\rm H\to\rm P}^{\rm B}=\frac{m-1}{m}(1+T_{\rm I\to\rm H}^{\rm B}+T_{\rm H\to\rm P}^{\rm B})+\frac{1}{m}(1+T_{\rm I\to\rm P}^{\rm B})\nonumber
	\end{equation}
	and
	\begin{equation}\label{S75}
	T_{\rm I\to\rm P}^{\rm B}=\frac{x_{\rm H}}{x_{\rm H}+x_{\rm P}}(1+T_{\rm H\to\rm P}^{\rm B})+\frac{x_{\rm P}}{x_{\rm H}+x_{\rm P}}.\nonumber
	\end{equation}
	Considering (\ref{S25}) and (\ref{S45}), 
	the above equations are resolved to obtain
	\begin{equation}\label{S76}
	T_{\rm H\to\rm I}^{\rm B}=4m+2\sqrt{2m-1}-1-\frac{2\sqrt{2m-1}}{m},
	\end{equation}
	\begin{equation}\label{S77}
	T_{\rm I\to\rm I}^{\rm B}=4m,
	\end{equation}
	\begin{equation}\label{S78}
	T_{\rm P\to\rm I}^{\rm B}=2m+1,
	\end{equation}
	\begin{small}
		\begin{equation}\label{S79}
		T_{\rm H\to\rm P}^{\rm B}=\frac{4m^2\sqrt{2m-1}+2m^3+4m^2-2m\sqrt{2m-1}-6m+2}{m\sqrt{2m-1}}
		\end{equation}
	\end{small}
	and
	\begin{equation}\label{S80}
	T_{\rm I\to\rm P}^{\rm B}=\frac{2m^2}{\sqrt{2m-1}}+2m-1.
	\end{equation}
	Substituting~\eqref{S45}, \eqref{S46}, and (\ref{S76})-(\ref{S80}) into (\ref{S59}) and \eqref{S53} yields \eqref{E5}.
	
	\emph{Case III: MERW}. 
	For MERW on $\mathcal{R}_m^4$, we can also build some relations among related hitting times:
	\begin{equation}\label{S82}
	T_{\rm H\to\rm I}^{\rm M}=\frac{m-1}{m}(1+T_{\rm I\to\rm H}^{\rm M}+T_{\rm H\to\rm I}^{\rm H})+\frac{1}{2m}(1+T_{\rm I\to\rm I}^{\rm M})+\frac{1}{2m},\nonumber
	\end{equation}
	\begin{equation}\label{S83}
	T_{\rm I\to\rm I}^{\rm M}=\frac{\mu_{\rm H}}{\mu_{\rm H}+\mu_{\rm P}}(1+T_{\rm H\to\rm I}^{\rm M})+\frac{\mu_{\rm P}}{\mu_{\rm H}+\mu_{\rm P}}(1+T_{\rm P\to\rm I}^{\rm M}),\nonumber
	\end{equation}
	\begin{equation}\label{S84}
	T_{\rm P\to\rm I}^{\rm M}=\frac{1}{2}(1+T_{\rm I\to\rm I}^{\rm M})+\frac{1}{2},\nonumber
	\end{equation}
	\begin{equation}\label{S85}
	T_{\rm H\to\rm P}^{\rm M}=\frac{m-1}{m}(1+T_{\rm I\to\rm H}^{\rm M}+T_{\rm H\to\rm P}^{\rm M})+\frac{1}{m}(1+T_{\rm I\to\rm P}^{\rm M})\nonumber
	\end{equation}
	and
	\begin{equation}\label{S86}
	T_{\rm I\to\rm P}^{\rm M}=\frac{\mu_{\rm H}}{\mu_{\rm H}+\mu_{\rm P}}(1+T_{\rm H\to\rm P}^{\rm M})+\frac{\mu_{\rm P}}{\mu_{\rm H}+\mu_{\rm P}}.\nonumber
	\end{equation}
	Making use of \eqref{S33} and  \eqref{S50}, 
	the above equations are solved to give
	\begin{equation}\label{S87}
	T_{\rm H\to\rm I}^{\rm M}=4m+1-\frac{2}{m},\nonumber
	\end{equation}
	\begin{equation}\label{S88}
	T_{\rm I\to\rm I}^{\rm M}=4m,\nonumber
	\end{equation}
	\begin{equation}\label{S89}
	T_{\rm P\to\rm I}^{\rm M}=2m+1,\nonumber
	\end{equation}
	\begin{equation}\label{S90}
	T_{\rm H\to\rm P}^{\rm M}=\frac{2(m+1)}{m}(m^2+m-1)\nonumber
	\end{equation}
	and
	\begin{equation}\label{S91}
	T_{\rm I\to\rm P}^{\rm M}=2m(m+1)-1.\nonumber
	\end{equation}
	Plugging the above-obtained results into (\ref{S59}) and \eqref{S53} results in \eqref{E5}.
	
	This completes the proof of the theorem.



\begin{thebibliography}{99}
	
	\bibitem{ChBa11}
	Chau, C.-K. and Basu, P. (2011) Analysis of latency of stateless opportunistic
	forwarding in intermittently connected networks.
	\newblock {\em IEEE/ACM Trans. Netw.}, {\bf  19}, 1111--1124.
	
	\bibitem{ZhYaTiGaWaXi15}
	Zheng, H., Yang, F., Tian, X., Gan, X., Wang, X., and Xiao, S. (2015) Data
	gathering with compressive sensing in wireless sensor networks: a random walk
	based approach.
	\newblock {\em IEEE Trans. Parallel Distrib. Syst.}, {\bf  26}, 35--44.
	
	\bibitem{LeKw16}
	Lee, C.-H. and Kwak, J. (2016) Towards distributed optimal movement strategy
	for data gathering in wireless sensor networks.
	\newblock {\em IEEE Trans. Parallel Distrib. Syst.}, {\bf  27}, 574--584.
	
	\bibitem{ElMaPr06}
	El~Gamal, A., Mammen, J., Prabhakar, B., and Shah, D. (2006) {Optimal
		throughput-delay scaling in wireless networks-part I: The fluid model}.
	\newblock {\em IEEE Trans. Inf. Theory}, {\bf  52}, 2568--2592.
	
	\bibitem{LiJiNiKa12}
	Liu, J., Jiang, X., Nishiyama, H., and Kato, N. (2012) Exact throughput
	capacity under power control in mobile ad hoc networks.
	\newblock {\em Proc. IEEE INFOCOM},  Orlando, Florida USA,  March 25-30,  pp.
	1--9. IEEE Press, Piscataway, NJ, USA.
	
	\bibitem{LiZh13IEEE}
	Li, Y. and Zhang, Z.-L. (2013) Random walks and green's function on digraphs: A
	framework for estimating wireless transmission costs.
	\newblock {\em IEEE/ACM Trans. Netw.}, {\bf  21}, 135--148.
	
	\bibitem{BeQuBa09}
	Beraldi, R., Querzoni, L., and Baldoni, R. (2009) Low hitting time random walks
	in wireless networks.
	\newblock {\em Wirel. Commun. Mob. Comput.}, {\bf  9}, 719--732.
	
	\bibitem{LiLiWaCh09}
	Lin, T., Lin, P., Wang, H., and Chen, C. (2009) Dynamic search algorithm in
	unstructured peer-to-peer networks.
	\newblock {\em IEEE Trans. Parallel Distrib. Syst.}, {\bf  20}, 654--666.
	
	\bibitem{PoLa05}
	Pons, P. and Latapy, M. (2005) Computing communities in large networks using
	random walks.
	\newblock {\em Proc. Int. Symp. Comput. Inform. Sci.},  Istanbul, Turkey,
	October 26 - 28,  pp. 284--293. Springer.
	
	\bibitem{FoPiReSa07}
	Fouss, F., Pirotte, A., Renders, J.-M., and Saerens, M. (2007) Random-walk
	computation of similarities between nodes of a graph with application to
	collaborative recommendation.
	\newblock {\em IEEE Trans. Knowl. Data Eng.}, {\bf  19}, 355--369.
	
	\bibitem{GoHuRa09}
	Gopalakrishnan, V., Hu, Y., and Rajan, D. (2009) Random walks on graphs to
	model saliency in images.
	\newblock {\em Proc. IEEE Int. Conf. Comput. Vision Pattern Recognit.},  Miami
	Beach, Florida USA,  June 20-25,  pp. 1698--1705. IEEE Press, Piscataway, NJ,
	USA.
	
	\bibitem{Gr06}
	Grady, L. (2006) Random walks for image segmentation.
	\newblock {\em IEEE Trans. Pattern Anal. Mach. Intell.}, {\bf  28}, 1768--1783.
	
	\bibitem{RiTo10}
	Ribeiro, B. and Towsley, D. (2010) Estimating and sampling graphs with
	multidimensional random walks.
	\newblock {\em Proc. ACM SIGCOMM IMC},  New Delhi, India,  August 30-September
	3,  pp. 390--403. ACM Press, New York, NY, USA.
	
	\bibitem{RiWaMuTo12}
	Ribeiro, B., Wang, P., Murai, F., and Towsley, D. (2012) Sampling directed
	graphs with random walks.
	\newblock {\em Proc. IEEE INFOCOM},  March 25-30,  Orlando, Florida USA,  pp.
	1692--1700. IEEE Press, Piscataway, NJ, USA.
	
	\bibitem{BrPa98}
	Brin, S. and Page, L. (1998) {The anatomy of a large-scale hypertextual Web
		search engine}.
	\newblock {\em Comput. Networks ISDN Syst.}, {\bf  30}, 107--117.
	
	\bibitem{YuZhTiTa14}
	Yu, J.-G., Zhao, J., Tian, J., and Tan, Y. (2014) Maximal entropy random walk
	for region-based visual saliency.
	\newblock {\em IEEE Trans. Cybern.}, {\bf  44}, 1661--1672.
	
	\bibitem{Ne03}
	Newman, M.~E. (2003) The structure and function of complex networks.
	\newblock {\em SIAM Rev.}, {\bf  45}, 167--256.
	
	\bibitem{Be09}
	Beraldi, R. (2009) Biased random walks in uniform wireless networks.
	\newblock {\em IEEE Trans. Mob. Comput.}, {\bf  8}, 500--513.
	
	\bibitem{GjKuBuMa10}
	Gjoka, M., Kurant, M., Butts, C.~T., and Markopoulou, A. (2010) Walking in
	facebook: A case study of unbiased sampling of osns.
	\newblock {\em Proc. IEEE INFOCOM},  San Diego, California USA,  March 14-19,
	pp. 1--9. IEEE Press, Piscataway, NJ, USA.
	
	\bibitem{IkKuYa09}
	Ikeda, S., Kubo, I., and Yamashita, M. (2009) The hitting and cover times of
	random walks on finite graphs using local degree information.
	\newblock {\em Theor. Comput. Sci.}, {\bf  410}, 94--100.
	
	\bibitem{MaBe11}
	Maiya, A.~S. and Berger-Wolf, T.~Y. (2011) Benefits of bias: Towards better
	characterization of network sampling.
	\newblock {\em Proc. Int. Conf. Knowl. Discov. Data Mining},  San Diego,
	California USA,  August 21-24,  pp. 105--113. ACM Press, New York, NY, USA.
	
	\bibitem{Pa64}
	Parry, W. (1964) Intrinsic {M}arkov chains.
	\newblock {\em Trans. Am. Math. Soc.}, {\bf  112}, 55--66.
	
	\bibitem{BuDuLuWa09}
	Burda, Z., Duda, J., Luck, J., and Waclaw, B. (2009) Localization of the
	maximal entropy random walk.
	\newblock {\em Phys. Rev. Lett.}, {\bf  102}, 160602.
	
	\bibitem{GoLa08}
	G\'omez-Garde\~nes, J. and Latora, V. (2008) Entropy rate of diffusion
	processes on complex networks.
	\newblock {\em Phys. Rev. E}, {\bf  78}, 065102.
	
	\bibitem{PeZh14}
	Peng, X. and Zhang, Z. (2014) Maximal entropy random walk improves efficiency
	of trapping in dendrimers.
	\newblock {\em J. Chem. Phys.}, {\bf  140}, 234104.
	
	\bibitem{LiZh14}
	Lin, Y. and Zhang, Z. ({2014}) {Mean first-passage time for maximal-entropy
		random walks in complex networks}.
	\newblock {\em Sci. Rep.}, {\bf  {4}}, 5365.
	
	\bibitem{KaGrTh13}
	Kafsi, M., Grossglauser, M., and Thiran, P. (2013) The entropy of conditional
	markov trajectories.
	\newblock {\em IEEE Trans. Inf. Theory}, {\bf  59}, 5577--5583.
	
	\bibitem{OcBu12}
	Ochab, J. and Burda, Z. (2012) Exact solution for statics and dynamics of
	maximal-entropy random walks on cayley trees.
	\newblock {\em Phys. Rev. E}, {\bf  85}, 021145.
	
	\bibitem{LiYuLi11}
	Li, R.-H., Yu, J.~X., and Liu, J. (2011) Link prediction: the power of maximal
	entropy random walk.
	\newblock {\em Proc. Int. Conf. Inform. Knowl. Manag.},  Glasgow, United
	Kingdom,  October 24-28,  pp. 1147--1156. ACM Press, New York, NY, USA.
	
	\bibitem{KoHu16}
	Korus, P. and Huang, J. (2016) Improved tampering localization in digital image
	forensics based on maximal entropy random walk.
	\newblock {\em IEEE Signal Process. Lett.}, {\bf  23}, 169--173.
	
	\bibitem{Bo72}
	Bonacich, P. (1972) Factoring and weighting approaches to status scores and
	clique identification.
	\newblock {\em J. Math. Sociol.}, {\bf  2}, 113--120.
	
	\bibitem{MaZhNe14}
	Martin, T., Zhang, X., and Newman, M. E.~J. (2014) Localization and centrality
	in networks.
	\newblock {\em Phys. Rev. E}, {\bf  90}, 052808.
	
	\bibitem{Ha89}
	Hashimoto, K. (1989) Zeta functions of finite graphs and representations of
	$p$-adic groups.
	\newblock {\em Adv. Stud. Pure Math.}, {\bf  15}, 211--280.
	
	\bibitem{KrMoMoNeSlZdZh13}
	Krzakala, F., Moore, C., Mossel, E., Neeman, J., Sly, A., Zdeborov®¢, L., and
	Zhang, P. (2013) Spectral redemption in clustering sparse networks.
	\newblock {\em Proc. Natl. Acad. Sci.}, {\bf  110}, 20935--20940.
	
	\bibitem{KaNeZd14}
	Karrer, B., Newman, M. E.~J., and Zdeborov\'a, L. (2014) Percolation on sparse
	networks.
	\newblock {\em Phys. Rev. Lett.}, {\bf  113}, 208702.
	
	\bibitem{LiChZh17}
	Lin, Y., Chen, W., and Zhang, Z. (2017) Assessing percolation threshold based
	on high-order non-backtracking matrices.
	\newblock {\em Proc. 26th Int. Conf. World Wide Web},  Perth, Australia,  April
	3-7,  pp. 223--232. International World Wide Web Conferences Steering
	Committee, Republic and Canton of Geneva, Switzerland.
	
	\bibitem{ShScMo15}
	Shrestha, M., Scarpino, S.~V., and Moore, C. (2015) Message-passing approach
	for recurrent-state epidemic models on networks.
	\newblock {\em Phys. Rev. E}, {\bf  92}, 022821.
	
	\bibitem{MoMa15}
	Morone, F. and Makse, H.~A. (2015) Influence maximization in complex networks
	through optimal percolation.
	\newblock {\em Nature}, {\bf  524}, 65--68.
	
	\bibitem{CoBeTeVoKl07}
	Condamin, S., B{\'e}nichou, O., Tejedor, V., Voituriez, R., and Klafter, J.
	(2007) First-passage times in complex scale-invariant media.
	\newblock {\em Nature}, {\bf  450}, 77--80.
	
	\bibitem{LiJuZh12}
	Lin, Y., Julaiti, A., and Zhang, Z.~Z. (2012) Mean first-passage time for
	random walks in general graphs with a deep trap.
	\newblock {\em J. Chem. Phys.}, {\bf  137}, 124104.
	
	\bibitem{ErGoSaSe14}
	Ermon, S., Gomes, C.~P., Sabharwal, A., and Selman, B. (2014) Designing fast
	absorbing {M}arkov chains.
	\newblock {\em Proc. AAAI},  July 27-31,  pp. 849--855. AAAI Press.
	
	\bibitem{WhSm03}
	White, S. and Smyth, P. (2003) Algorithms for estimating relative importance in
	networks.
	\newblock {\em Proc. ACM SIGKDD Int. Conf. Knowl. Discov. Data Mining},
	Washington DC, USA,  August 24-27,  pp. 266--275. ACM Press, New York, NY,
	USA.
	
	\bibitem{FeQuYi14}
	Feng, M., Qu, H., and Yi, Z. (2014) Highest degree likelihood search algorithm
	using a state transition matrix for complex networks.
	\newblock {\em {IEEE Trans. Circuits and Syst. I, Reg. Papers}}, {\bf  61},
	2941--2950.
	
	\bibitem{Lo96}
	Lov{\'a}sz, L. (1993) Random walks on graphs: A survey.
	\newblock {\em Combinatorics, Paul Erd\"os is Eighty}, {\bf  2}, 1--46.
	
	\bibitem{AlFi99}
	Aldous, D. and Fill, J. (1999) {\em Reversible Markov chains and random walks
		on graphs}.
	\newblock http://www.stat.berkeley.edu/$\sim$aldous/RWG/book.html.
	
	\bibitem{St09}
	Strang, G. (2009) {\em Introduction to Linear Algebra}. Wellesley-Cambridge
	Press, Wellesley, MA.
	
	\bibitem{Ba92}
	Bass, H. (1992) {The Ihara-Selberg zeta function of a tree lattice}.
	\newblock {\em Int. J. Math.}, {\bf  03}, 717--797.
	
	\bibitem{AnFrHo07}
	Angel, O., Friedman, J., and Hoory, S. (2015) The non-backtracking spectrum of
	the universal cover of a graph.
	\newblock {\em Trans. Am. Math. Soc.}, {\bf  367}, 4287--4318.
	
	\bibitem{ZhShCh13}
	Zhang, Z., Shan, T., and Chen, G. (2013) Random walks on weighted networks.
	\newblock {\em Phys. Rev. E}, {\bf  87}, 012112.
	
	\bibitem{BaBaPaVe04}
	Barrat, A., Barthelemy, M., Pastor-Satorras, R., and Vespignani, A. (2004) The
	architecture of complex weighted networks.
	\newblock {\em Proc. Natl. Acad. Sci. U.S.A.}, {\bf  101}, 3747--3752.
	
	\bibitem{AlBeLu07}
	Alon, N., Benjamini, I., Lubetzky, E., and Sodin, S. (2007) Non-backtracking
	random walks mix faster.
	\newblock {\em Commun. Contemp. Math.}, {\bf  9}, 585--603.
	
	\bibitem{FiHo13}
	Fitzner, R. and van~der Hofstad, R. (2013) Non-backtracking random walk.
	\newblock {\em J. Stat. Phys.}, {\bf  150}, 264--284.
	
	\bibitem{Ke16}
	Kempton, M. (2016) Non-backtracking random walks and a weighted {I}hara's
	theorem.
	\newblock {\em Open J. Discrete Math.}, {\bf  6}, 207--226.
	
	\bibitem{LiZh13}
	Lin, Y. and Zhang, Z.~Z. (2013) {Random walks in weighted networks with a
		perfect trap: An application of Laplacian spectra}.
	\newblock {\em Phys. Rev. E}, {\bf  87}, 062140.
	
	\bibitem{ErRe60}
	Erd\"os, P. and R{\'e}nyi, A. (1960) On the evolution of random graphs.
	\newblock {\em Publ. Math. Inst. Hungar. Acad. Sci}, {\bf  5}, 17--61.
	
	\bibitem{BaAl99}
	Barab{\'a}si, A.-L. and Albert, R. (1999) Emergence of scaling in random
	networks.
	\newblock {\em Science}, {\bf  286}, 509--512.
	
	\bibitem{BeDe70}
	Bell, R. and Dean, P. (1970) Atomic vibrations in vitreous silica.
	\newblock {\em Discuss. Faraday Soc.}, {\bf  50}, 55--61.
	
	\bibitem{WaSt98}
	Watts, D.~J. and Strogatz, S.~H. (1998) Collective dynamics of
	`small-world'networks.
	\newblock {\em Nature}, {\bf  393}, 440--442.
	
	\bibitem{LuScBoHaSlDa03}
	Lusseau, D., Schneider, K., Boisseau, O.~J., Haase, P., Slooten, E., and
	Dawson, S.~M. (2003) The bottlenose dolphin community of doubtful sound
	features a large proportion of long-lasting associations.
	\newblock {\em Behav. Ecol. Sociobiol.}, {\bf  54}, 396--405.
	
	\bibitem{Ne06}
	Newman, M. E.~J. (2006) Finding community structure in networks using the
	eigenvectors of matrices.
	\newblock {\em Phys. Rev. E}, {\bf  74}, 036104.
	
	\bibitem{DuAr05}
	Duch, J. and Arenas, A. (2005) Community detection in complex networks using
	extremal optimization.
	\newblock {\em Phys. Rev. E}, {\bf  72}, 027104.
	
	\bibitem{GuDaDiGiAr03}
	Guimera, R., Danon, L., Diaz-Guilera, A., Giralt, F., and Arenas, A. (2003)
	Self-similar community structure in a network of human interactions.
	\newblock {\em Phys. Rev. E}, {\bf  68}, 065103.
	
	\bibitem{LeKlFa07}
	Leskovec, J., Kleinberg, J., and Faloutsos, C. (2007) Graph evolution:
	Densification and shrinking diameters.
	\newblock {\em ACM Trans. Knowl. Discov. Data}, {\bf  1}, 2.
	
	\bibitem{TeBeVo09}
	Tejedor, V., B{\'e}nichou, O., and Voituriez, R. (2009) Global mean
	first-passage times of random walks on complex networks.
	\newblock {\em Phys. Rev. E}, {\bf  80}, 065104.
	
	\bibitem{LiHu13}
	Liu, F. and Huang, Q. (2013) Laplacian spectral characterization of 3-rose
	graphs.
	\newblock {\em Linear Algebra Appl.}, {\bf  439}, 2914--2920.
	
\end{thebibliography}

\end{document}